\documentclass[onecolumn, journal]{IEEEtran}

\usepackage{graphicx}
\usepackage{latexsym}

\usepackage{amssymb,amsmath,amsthm}
\usepackage{mathrsfs}

\usepackage[T1]{fontenc}
\usepackage{amsfonts}
\usepackage{bm}
\usepackage{dsfont}
\usepackage{color}
\usepackage[mathscr]{eucal}
\usepackage{cite}
\usepackage{color}
\usepackage{balance}
\usepackage{url}
\usepackage{bbm}
\usepackage{subfigure}
\usepackage{changepage}
\usepackage{multirow} 

\usepackage{tabularx}

\usepackage{array}
\newcolumntype{V}{>{\centering\arraybackslash} m{.4\linewidth} }

\newcommand{\prob}[1]{\mathsf{Pr}\left(#1\right)}
\newcommand{\GacsKorner}{G\'{a}cs-K\"{o}rner\ }
\newcommand{\Gacs}{G\'{a}cs\ }

\DeclareMathOperator{\Ent}{Ent}
\DeclareMathOperator{\Ber}{Ber}
\DeclareMathOperator{\Var}{Var}
\DeclareMathOperator{\DSBS}{DSBS}
\DeclareMathOperator{\DSBStriple}{DSBS-triple}
\newcommand{\indicator}[1]{\mathds{1}_{\left[ {#1} \right] }}

\newtheorem{theorem}{Theorem}
\newtheorem{lemma}{Lemma}

\newtheorem{corollary}{Corollary}

\newtheorem{observation}{Observation}

\theoremstyle{definition}
\newtheorem{definition}{Definition}

\theoremstyle{remark}
\newtheorem{remark}{Remark}
\newtheorem{example}{Example}

\IEEEoverridecommandlockouts

\title{On Non-Interactive Simulation of Joint Distributions}

\begin{document}

\title{On Non-Interactive Simulation of Joint Distributions} \author{
  \IEEEauthorblockN{Sudeep Kamath\IEEEauthorrefmark{1}, Venkat
    Anantharam\IEEEauthorrefmark{2}\thanks{Part of this paper was
      presented at the 50th Annual Allerton Conference on
      Communications, Control and Computing 2012, Monticello,
      Illinois. This document is the final version of the paper to
      appear in the IEEE Transactions on Information Theory.}
  }\\
  \IEEEauthorblockA{\IEEEauthorrefmark{1}ECE Department, Princeton
    University,
    \\ sukamath@princeton.edu}\\
  \IEEEauthorblockA{\IEEEauthorrefmark{2}EECS Department, University
    of California, Berkeley, \\ ananth@eecs.berkeley.edu} }
\date{\today}
\maketitle
\thispagestyle{plain}
\pagestyle{plain}

\begin{abstract}
  We consider the following non-interactive simulation problem: Alice
  and Bob observe sequences $X^n$ and $Y^n$ respectively where
  $\{(X_i, Y_i)\}_{i=1}^n$ are drawn i.i.d. from $P(x,y),$ and they
  output $U$ and $V$ respectively which is required to have a joint
  law that is close in total variation to a specified $Q(u,v).$ It is
  known that the maximal correlation of $U$ and $V$ must necessarily
  be no bigger than that of $X$ and $Y$ if this is to be possible. Our
  main contribution is to bring hypercontractivity to bear as a tool
  on this problem. In particular, we show that if $P(x,y)$ is the
  doubly symmetric binary source, then hypercontractivity provides
  stronger impossibility results than maximal correlation. Finally, we
  extend these tools to provide impossibility results for the
  $k$-agent version of this problem.

\end{abstract}

\section{Introduction}

The problem of simulating random variables by two agents with suitable
resource constraints has had a rich history leading to different
formulations of this problem in the literature. The general setup for
the problem is as follows: Two or more agents wish to simulate a
specified joint distribution under resource constraints in the form of
limited communication, limited common randomness provided to all of
them, or limited correlation between their
observations. One then wishes to find the minimum resources required
to achieve the desired goal.

The simulation problem has natural applications in numerous areas ---
from game-theoretic co-ordination in a network against an adversary to
control of a dynamical system over a distributed network. These
problems are expected to be important in many future technologies with
remote-controlled applications, such as Amazon's drone-based delivery
system \cite{AmazonPrimeAir} and robotic environmental cleanup,
vegetation management, land clearing, and bio-mass harvesting
\cite{RoboticsEngineering}. In these technologies, individual robotic
components would need to take randomized actions under limited or no
communication with other components or the central system. Study of
the simulation problem can provide fundamental limits on the
capabilities of such robotic components and guide efficient usage of
the available resources.

The earliest studied two-agent simulation problems were considered by
G\'{a}cs and K\"{o}rner \cite{GacsKorner72}, and Wyner
\cite{Wyner75}. One may interpret their results, which we will
describe shortly, in the framework of a generalization of both their
problem setups as shown in Fig.~\ref{fig:GacsKornerWyner}. Let the
random variables $X,Y,U,V$ shown take values in finite sets.

\begin{figure}[h]
  \begin{center}
    \includegraphics[width = 3in, height=!]{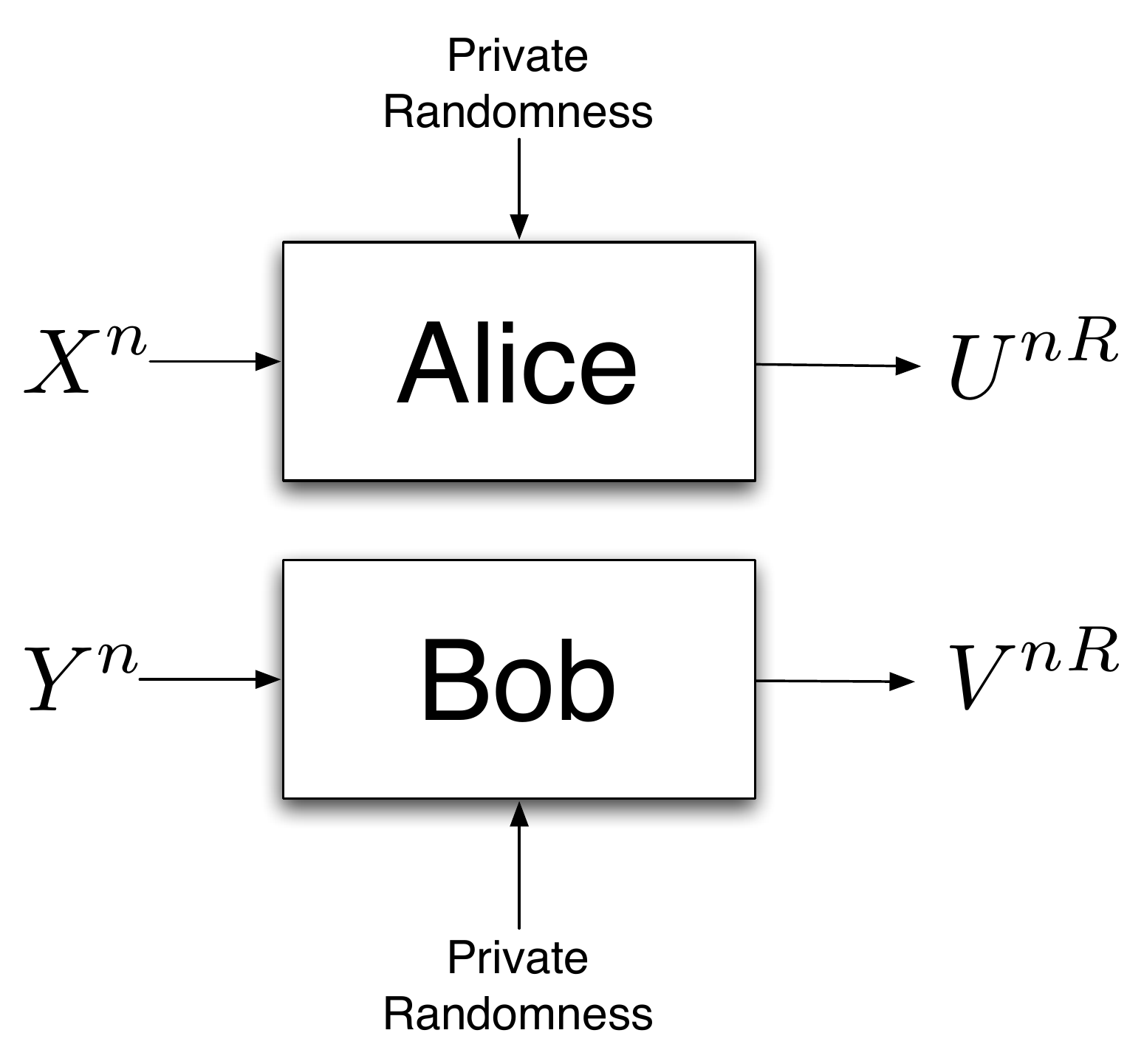}
    \caption{A generalization of the problem setups considered by
      \GacsKorner \cite{GacsKorner72} and Wyner\cite{Wyner75}}
    \label{fig:GacsKornerWyner}
  \end{center}
\end{figure}

In this formulation, two agents each having access to its own infinite
stream of private randomness, observe $n$ i.i.d. copies of samples
generated according to a specified law $P(x,y)$ as shown, and are
required to output $nR$ samples drawn from a distribution that is
close (in total variation) to the the distribution constructed by
taking i.i.d. copies of a specified law $Q(u,v).$ Let the simulation
capacity $R^*$ be defined as the supremum of all rates for which given
any $\epsilon>0,$ it is possible for some $n$ to carry out this task
to within total variation distance $\epsilon$.
\begin{itemize}
\item When $Q(u,v)$ is described by $U=V\sim\Ber(1/2),$ and $P(x,y)$
  is a general distribution, this problem considers fundamental limits
  for \emph{extracting} common randomness from the distribution of
  $(X,Y).$ G\'{a}cs and K\"{o}rner showed in \cite{GacsKorner72} that
  we have the \emph{simulation capacity} $R^* = K(X;Y),$ which has
  come to be known as the G\'{a}cs-K\"{o}rner common information of
  $X$ and $Y.$ This quantity $K(X;Y)$ can be described as $\sup
  H(\Theta)$ where $\Theta=f(X)=g(Y).$ In other words, the simulation
  capacity is non-zero only when the distribution of $(X,Y)$ is
  \emph{decomposable}, i.e. $\mathcal{X}$ may be partitioned as
  $\mathcal{X}_1\cup\mathcal{X}_2$ and $\mathcal{Y}$ may be
  partitioned as $\mathcal{Y}_1\cup\mathcal{Y}_2$ so that $\prob{X\in
    \mathcal{X}_1,Y\in \mathcal{Y}_2} = \prob{X\in
    \mathcal{X}_2,Y\in\mathcal{Y}_1}=0$ and $\prob{X\in
    \mathcal{X}_1,Y\in \mathcal{Y}_1}, \prob{X\in
    \mathcal{X}_2,Y\in\mathcal{Y}_2}>0.$ Further, they showed that in
  general, $K(X;Y)\leq I(X;Y).$
\item When $P(x,y)$ is described by $X=Y\sim\Ber(1/2),$ and $Q(u,v)$
  is a general distribution, this problem considers fundamental limits
  for common randomness needed for \emph{generating} the random
  variable pair $(U,V)$. Wyner showed in \cite{Wyner75} that the
  amount of common information needed for generation per sample is
  $(R^*)^{-1} = C(U;V),$ which has come to be known as the Wyner
  common information of $U$ and $V.$ This quantity $C(U;V)$ can be
  described as $\sup I(\Theta;U,V)$ over all $\Theta$ satisfying
  $U-\Theta-V$ with cardinality bound on the variable $\Theta$ given
  by $|\mathsf{\Theta}|\leq |\mathcal{U}|\cdot |\mathcal{V}|.$
  Further, Wyner showed that $C(U;V)\geq I(U;V)$ in general. To be
  precise, Wyner considered a problem setting that required $(U,V)$ to
  be simulated with vanishing normalized relative entropy, i.e. if
  $Q^\prime(u^{nR},v^{nR})$ is the law of the simulated samples, and
  $Q(u,v)$ was the target distribution, then simulation is considered
  possible in Wyner's formulation if
  \begin{align}
    \frac{1}{nR}D\left(Q^\prime(u^{nR},v^{nR})||\Pi_{i=1}^{nR}Q(u_i,v_i)\right)\to 0.
  \end{align}
  It has been recognized that the simulation capacity remains the same
  under the vanishing total variation constraint
  \cite[Lemma~5]{HanVerdu93}, \cite[Lemma~IV.1]{Cuff13}. A recent work
  \cite{KumarLiElGamal14} considers a variant of Wyner's problem with
  \emph{exact} generation of random variables as opposed to generation
  with a vanishing total variation distance.
\end{itemize}

The problem of characterizing $R^*$ is open for general distributions
$P(x,y)$ and $Q(u,v),$ and so is the problem of characterizing when $R^* >
0.$

In another stream of related work, the problem of simulation has been
considered under rate-limited interaction between the
agents. This began with the work of Cuff \cite{Cuff08} who studied
communication requirements for simulating a channel with rate-limited
communication and rate-limited common randomness. \cite{Cuff10}
studied communication requirements for establishing dependence among
nodes in a network setting. The former setup (of Cuff \cite{Cuff08})
was generalized by Gohari and Anantharam in \cite{GohariAnantharam11}
(see Fig.~\ref{fig:GohariAnantharam}). Two agents wish to simulate
i.i.d. samples of a specified joint distribution $P(x,y,u,v).$ Nature
supplies i.i.d. copies of $(X,Y)$ with the right marginal distribution
as shown and the agents can use a certain rate of common randomness,
certain rate-limited communication, and infinite streams of individual
private randomness to accomplish the desired task. We want to
understand the fundamental trade-offs between these rates to make this
task possible. This problem was completely solved by Yassaee, Gohari,
and Aref in \cite{Yassaee12}. However, this work does not address the
problem of computing the simulation capacity $R^*$ for the setup in
Fig.~\ref{fig:GacsKornerWyner}, since the problem formulation there is
different in two respects: In Fig.~\ref{fig:GohariAnantharam}, the
task is to output $n$ samples while in Fig.~\ref{fig:GacsKornerWyner},
the task is to output $nR$ samples. Furthermore, even if $R$ were say
chosen to be 1, in Fig.~\ref{fig:GohariAnantharam}, the joint
distribution of the quadruple $(X^n,Y^n,U^n,V^n)$ is required to be
close to i.i.d. copies of a specified joint distribution. However, in
Fig.~\ref{fig:GacsKornerWyner}, the requirement is only on the
marginal distribution of the output samples $(U^n,V^n)$ and the
quadruple $(X^n,Y^n,U^n,V^n)$ need not even be close to an
i.i.d. distribution.

\begin{figure}[h]
  \begin{center}
    \includegraphics[width = 3in, height=!]{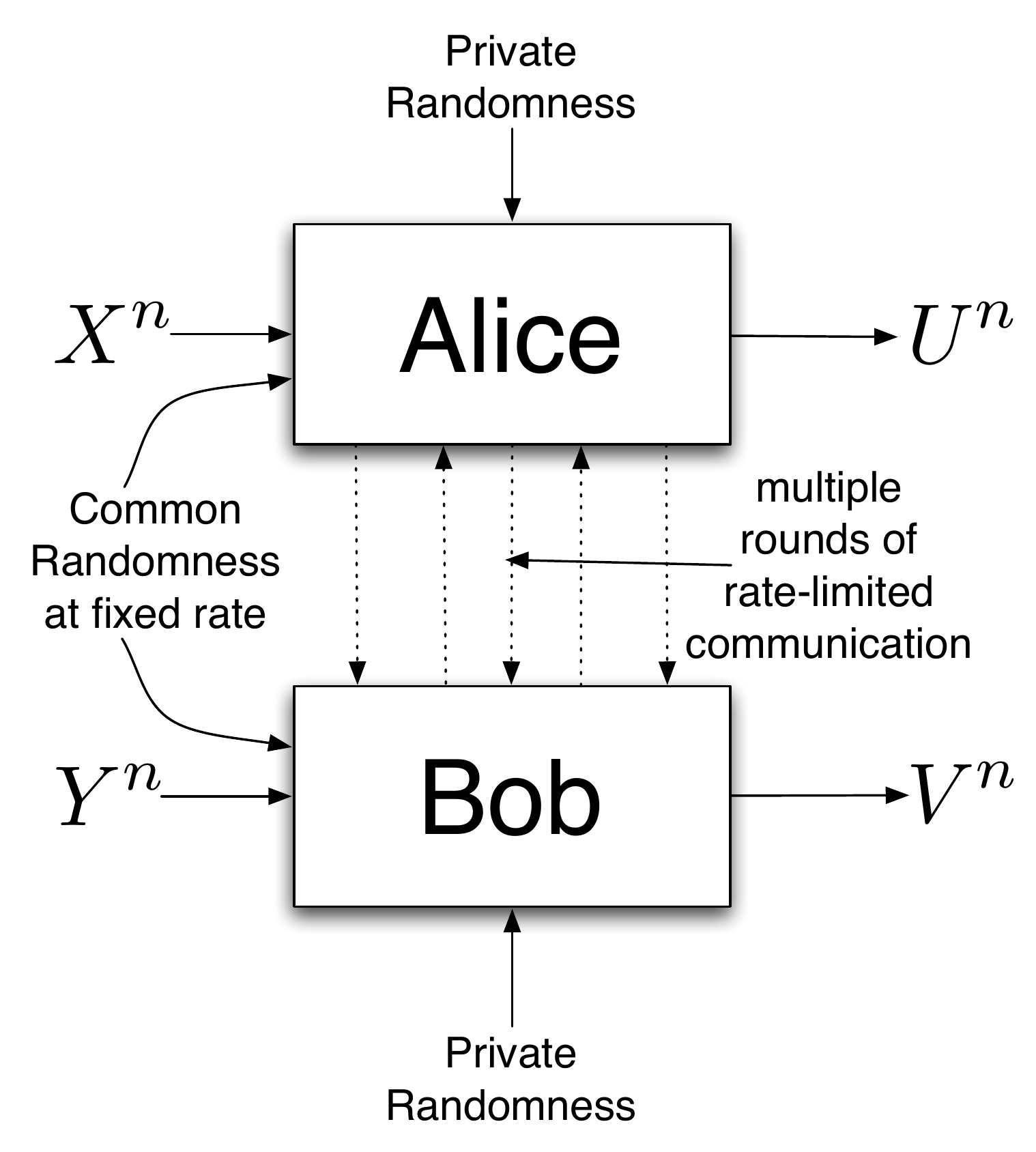}
    \caption{Generalization of Cuff's formulation \cite{Cuff08} by
      Gohari and Anantharam \cite{GohariAnantharam11}}
    \label{fig:GohariAnantharam}
  \end{center}
\end{figure}

In this paper, we consider the former non-interactive simulation setup
\`{a} la \GacsKorner and Wyner (Fig.~\ref{fig:GacsKornerWyner}).
Since the problem of characterizing whether $R^* > 0$ for general
distributions $P(x,y)$ and $Q(u,v),$ is also non-trivial, we propose a
relaxed problem where two agents observe an arbitrary finite number of
samples drawn i.i.d. from $P(x,y)$ as shown in
Fig.~\ref{fig:non_interactive_simulation} and are required to output
one random variable each with the requirement that the output
distribution be close in total variation to a specified $Q(u,v).$
Clearly, if it is impossible to generate even a single sample, we must
have $R^*=0.$ We therefore focus on impossibility results for this
problem which will be relevant to the formulation in
Fig.~\ref{fig:GacsKornerWyner}. It is not clear if the converse is
true, i.e. it is unclear whether the feasibility of generating one
sample asymptotically implies that we may generate samples at a rate
$R>0.$

Note that the notion of \textit{simulation} we consider is
distinct from the notion of \textit{exact generation} wherein a
certain distribution is required to be generated exactly. If we have a
strategic setting, such as a distributed game, in which a player,
represented by a number of distributed agents, is playing against an
adversary, the agents would often need to generate a joint
distribution exactly \cite{AnantharamBorkar07}, to avoid providing
unforeseen strategic advantages to the adversary.

\begin{figure}[h]
  \begin{center}
    \includegraphics[width = 3in, height=!]{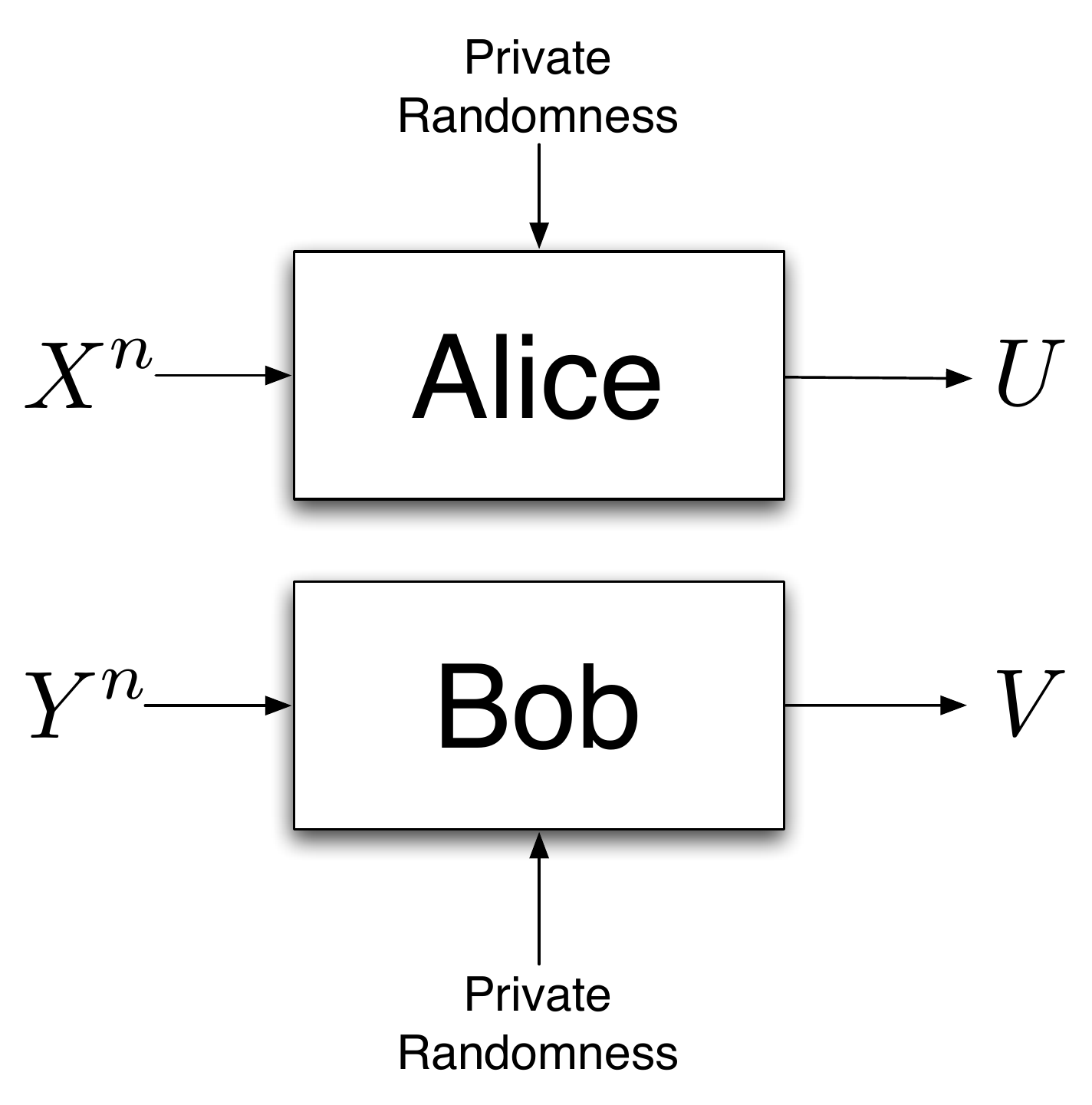}
    \caption{The non-interactive simulation problem considered in this
    paper}
    \label{fig:non_interactive_simulation}
  \end{center}
\end{figure}

When $(U,V)\sim Q(u,v)$ is described by $U=V\sim\Ber(1/2)$ while
$P(x,y)$ is a general distribution, the problem has recently come to
be called \emph{non-interactive correlation distillation}
\cite{Mossel04, Mossel10}. We therefore, call our formulation the
problem of \emph{non-interactive simulation of joint
  distributions}. In a remarkable strengthening of the
G\'{a}cs-K\"{o}rner result \cite{GacsKorner72}, Witsenhausen showed in
\cite{Witsenhausen75} that unless the G\'{a}cs-K\"{o}rner common
information $K(X;Y)$ is positive (i.e. the joint distribution of
$(X,Y)$ is decomposable), non-interactive correlation distillation is
impossible to achieve. The chief tool used in Witsenhausen's proof is
the \emph{maximal correlation} of two random variables, a quantity
which will be of prime importance in the present paper as well.

The second tool that we will be using is \emph{hypercontractivity},
which has found numerous applications in mathematics, physics, and
theoretical computer science. The origins of hypercontractivity lie in
the early works of Bonami \cite{Bonami68,Bonami70}, of Nelson
\cite{Nelson} in quantum field theory, of Gross \cite{Gross} who first
developed the connection to logarithmic Sobolev inequalities, and of
Beckner \cite{Beckner}. The meaning of hypercontractivity was
broadened by Borell \cite{Borell} to what is sometimes called reverse
hypercontractivity today \cite{Mossel11}. Hypercontractivity has found
powerful applications in a lot of fields, for example the study of
influence of variables on Boolean functions \cite{KKL88},
\cite{Friedgut98}, \cite{MOO05} and in voting system
theory\cite{MosselRacz11}. Ahlswede and \Gacs \cite{AhlswedeGacs76}
identified the use of hypercontractivity in studying the spreading of
sets in high dimensional product spaces. In recent works,
\cite{Raginsky13} showed an equivalence between hypercontractivity and
strong data processing inequalities for R\'{e}nyi divergences,
\cite{Polyanskiy13a} used hypercontractivity to show non-vanishing
lower bounds on hypothesis testing, \cite{Polyanskiy13b} studied
hypercontractivity for a noise operator that computed spherical
averages in Hamming space, \cite{AGKN_14} showed a connection between
hypercontractivity and strong data processing inequalities for mutual
information, and \cite{AGKN_Allerton13} used hypercontractivity to
study the mutual information between Boolean functions. As we shall
see, hypercontractivity has properties that make it naturally
well-suited for studying the non-interactive simulation problem.


Let us formally set up the non-interactive simulation problem
described earlier.

\begin{definition}
  Let $\mathcal{X}, \mathcal{Y}, \mathcal{U}, \mathcal{V}$ denote
  finite sets. Given a \textit{source distribution} $P(x,y)$ over
  $\mathcal{X}\times \mathcal{Y}$ and a \textit{target distribution}
  $Q(u,v)$ over $\mathcal{U}\times\mathcal{V},$ we say that
  \emph{non-interactive simulation} of $Q(u,v)$ using $P(x,y)$ is
  possible, if for any $\epsilon>0,$ there exists a positive integer
  $n,$ a finite set $\mathcal{R},$ and functions
  $f:\mathcal{X}^n\times \mathcal{R}\mapsto \mathcal{U},$\ 
  $g:\mathcal{Y}^n\times\mathcal{R}\mapsto \mathcal{V}$ such that
$$d_{\mathrm{TV}}\left((f(X^n,M_X),g(Y^n,M_Y));(U,V)\right)\leq \epsilon$$ where
$\{(X_i,Y_i)\}_{i=1}^n$ is a sequence of i.i.d. samples drawn from
$P(x,y),$ $M_X,M_Y$ are uniformly distributed in $\mathcal{R}$ and are
mutually independent of each other and the samples from the source,
$(U,V)$ is drawn from $Q(u,v)$ and $d_{\mathrm{TV}}(\cdot\,;\cdot)$ is
the total variation distance (defined as half the $L_1$ distance
between the distributions).
\end{definition}

For a fixed $P(x,y),$ the set of distributions $Q(u,v)$ on a fixed set
$\mathcal{U}\times\mathcal{V}$ for which non-interactive simulation is
possible is precisely the closure of the set of marginal distributions
of $(U,V)$ satisfying $U-X^k-Y^k-V$ for some $k.$ However, this set of
distributions appears to be very hard to characterize explicitly. In
this paper, we focus on outer bounds on this set, or in other words
impossibility results for non-interactive simulation.

Note that since we are interested only in determining the possibility
of simulation and not in the simulation capacity, the problem does not
have any less generality if we disallow the agents from using any
private randomness, since agents can obtain as much private randomness
as desired by using extended observations that are non-overlapping in
time, i.e. the agents observe $n_1+n_2+n_3$ symbols, they use
$(X_{1},\ldots, X_{n_1}), (Y_{1},\ldots, Y_{n_1})$ respectively as
their correlated observations, Alice uses $X_{n_1+1},\ldots, X_{n_2}$
as her private randomness, and Bob uses $Y_{n_2+1},\ldots, Y_{n_3}$ as
his private randomness. We make the choice to assume the availability
of private randomness as part of the model.

We will consider two examples to motivate the focus of this
study.

\subsection{Example 1}
\label{subsec:example1}
Let $X$ be a uniform Bernoulli random variable,
$X\sim\Ber(\frac{1}{2}).$ Let $Y$ be a noisy copy of $X,$ i.e. $Y=X+N$
where $N\sim\Ber(\alpha)$ for $0<\alpha<\frac{1}{2},$ is independent
of $X.$ Here, the addition is modulo 2. We say that $(X,Y)$ has the
\emph{doubly symmetric binary source} distribution with parameter
$\alpha,$ denoted $\DSBS(\alpha)$ following the notation of Wyner
\cite{Wyner75}. We consider $(U,V)\sim\DSBS(\beta)$ for $0\leq
\beta<\frac{1}{2}.$ We may ask whether non-interactive simulation of
$Q(u,v) = \DSBS(\beta)$ using $P(x,y) = \DSBS(\alpha)$ is
possible. Witsenhausen answered this question in the negative when
$\beta<\alpha$ in \cite{Witsenhausen75}, thus significantly
strengthening the result of G\'{a}cs and K\"{o}rner
\cite{GacsKorner72}. Witsenhausen established this by proving the
tensorization of the maximal correlation of an arbitrary pair of
random variables (both tensorization and maximal correlation are
defined and discussed in Section~\ref{subsec:mc}).  This can be used
to conclude that if non-interactive simulation is possible, then the
maximal correlation of the target distribution can be no more than
that of the source distribution. The parameter $n$ has disappeared in
this comparison thanks to the tensorization property. The maximal
correlation of a pair of binary random variables distributed as
$\DSBS(\alpha)$ equals $|1-2\alpha|.$ Thus, for instance, if the
non-interactive simulation of $\DSBS(\beta)$ using $\DSBS(\alpha)$ is
possible, with $0\leq \alpha,\beta\leq \frac{1}{2},$ then we must have
$\alpha\leq \beta.$ Furthermore, it is easy to see that if $\alpha\leq
\beta,$ then non-interactive simulation is indeed possible: Alice
outputs the first bit of her observation while Bob outputs a suitable
noisy copy of his first bit. Thus, for $0\leq \alpha,\beta\leq
\frac{1}{2},$ non-interactive simulation of $\DSBS(\beta)$ using
$\DSBS(\alpha)$ is possible if and only if $\alpha\leq \beta.$

\subsection{Example 2}
\label{subsec:example2}
Let $P(x,y)$ be given by $(X,Y)\sim \DSBS(\alpha)$ with
$0<\alpha<\frac{1}{2}.$ Consider binary random variables $(U,V)$
distributed as $Q(u,v)$ given by: $Q(0,0)=0, Q(0,1) = Q(1,0) = Q(1,1)
= \frac{1}{3}.$ We ask if non-interactive simulation of $Q(u,v)$ using
$\DSBS(\alpha)$ is possible. The maximal correlation of a
$\DSBS(\alpha)$ source distribution is $|1-2\alpha|$ while that of
$Q(u,v)$ is $\frac{1}{2}.$ Since non-interactive simulation is
impossible unless the maximal correlation of the source exceeds that of
the target, we have non-interactive simulation impossible if
$|1-2\alpha|\leq \frac{1}{2},$ i.e. $\frac{1}{4}<\alpha<\frac{1}{2}.$
But what about the case when $0<\alpha\leq \frac{1}{4}?$ Can we come
up with a suitable scheme to simulate $Q(u,v)$? The answer turns out
to be no for each $0<\alpha\leq \frac{1}{4}$ and can be proved using the
following inequality which holds for $\{(X_i, Y_i)\}_{i=1}^n$ being
i.i.d. $\DSBS(\alpha),$ and for arbitrary sets $S,T\subseteq\{0,1\}^n:$ 
\begin{equation}\label{eqn:basic-reverse-hc}
\prob{X^n\in S, Y^n\in T}\geq \prob{X^n\in S}^{\frac{1}{2\alpha}}\prob{Y^n\in T}^{\frac{1}{2\alpha}}.
\end{equation}
The above inequality follows from a so-called reverse hypercontractive
inequality \cite[Thm. 3.4]{Mossel04}. We will revisit this inequality
in Section~\ref{subsec:hc-ribbon}. If non-interactive simulation of
$Q(u,v)$ using $\DSBS(\alpha)$ were possible, we should be able to
find sets $S,T$ such that
$\prob{X^n\in S}\approx \frac{1}{3}, \prob{Y^n\in T}\approx
\frac{1}{3}$
and $\prob{X^n\in S, Y^n\in T}\approx 0.$ Inequality
\eqref{eqn:basic-reverse-hc} rules out this possibility (assuming
private randomness is not available, which we had argued is without
loss of generality). Thus, hypercontractivity or reverse
hypercontractivity can provide impossibility results when the maximal
correlation approach cannot.  Is it true that one is always stronger
than the other? One of the main results in our paper is that
hypercontractivity allows for stronger impossibility results than the
maximal correlation when $P(x,y) = \DSBS(\alpha).$ More generally, we
give necessary and sufficient conditions on $P(x,y)$ for this
subsumption. This arises from an inequality obtained by Ahlswede and
G\'{a}cs\cite{AhlswedeGacs76} in the hypercontractive case which we
extend to the reverse hypercontractive case.

The rest of the paper is organized as follows. Section~\ref{sec:mc-hc}
discusses preliminaries on maximal correlation and
hypercontractivity. We present our main results in
Section~\ref{sec:main-results}. As mentioned earlier, one of our main
results is a necessary and sufficient condition on the source
distribution $P(x,y)$ which allows one to definitively conclude that
hypercontractivity will provide stronger impossibility results than
maximal correlation. As our second main result, we give a
characterization of a limiting hypercontractivity parameter (that we
call $s^*$) as a strong data processing constant for KL
divergences. This characterization was first proven by Ahlswede-Gacs
\cite{AhlswedeGacs76}. However, our proof has the advantage of being
more intuitive - arising naturally from a Taylor series expansion -
while at the same time extending immediately to reverse
hypercontractivity.  This hypercontractivity parameter has recently
been shown to also be the tightest constant in strong data processing
inequalities for mutual information
\cite{AGKN_14}. Section~\ref{sec:triple} discusses the extension of
the non-interactive simulation problem for $k\geq 3$ agents. We
provide a couple of interesting three-user non-interactive simulation
examples where every two agents can simulate the corresponding
pairwise marginal of the desired joint distribution but the triple
cannot simulate the triple joint distribution.

\section{Main Tools: Maximal Correlation and Hypercontractivity}
\label{sec:mc-hc}
In this paper, all sets are finite and all probability distributions
are discrete and have finite support. We denote the marginals of
$P(x,y)$ and $Q(u,v)$ by $P_X(x), P_Y(y)$ and $Q_U(u), Q_V(v)$
respectively. We will use $\mathbb{R}_{\geq 0}$ and $\mathbb{R}_{>0}$
to denote non-negative reals and strictly positive reals respectively.
In the following subsections, we will review the definition and
properties of maximal correlation and hypercontractivity.

\subsection{Maximal Correlation}
\label{subsec:mc}

For jointly distributed random variables $(X,Y),$ define their
\emph{maximal correlation} $\rho_m(X;Y):=\sup \mathbb{E}f(X)g(Y)$ where
the supremum is taken over $f:\mathcal{X}\mapsto\mathbb{R},
g:\mathcal{Y}\mapsto\mathbb{R}$ such that
$\mathbb{E}f(X)=\mathbb{E}g(Y)=0$ and
$\mathbb{E}f(X)^2,\mathbb{E}g(Y)^2\leq 1.$

\begin{example}
 If $(X,Y)\sim\DSBS(\alpha),$ then the only functions $f,g$ satisfying the
 conditions $\mathbb{E}f(X)=\mathbb{E}g(Y)=0$ and
$\mathbb{E}f(X)^2,\mathbb{E}g(Y)^2\leq 1$ are $f(x) = a(1_{x=0} -
1_{x=1})$ and $g(y) = b(1_{y=0} -1_{y=1})$ with $|a|,|b|\leq 1.$ The
optimum is then achieved with $a=b=1$ if $\alpha<\frac 12$ and with
$a=b=-1$ if $\alpha\geq \frac 12.$ Thus, 

 \begin{equation}\rho_m(X;Y)=|1-2\alpha|.\label{eqn:DSBS-mc}\end{equation}
\end{example}

The following properties of the maximal correlation of two discrete
random variables with finite support can be shown easily
\cite{Renyi59}.
\begin{enumerate}
\item $0\leq \rho_m(X;Y)\leq 1.$
\item $\rho_m(X;Y)=0$ if and only if $X$ is independent of $Y.$
\item $\rho_m(X;Y)=1$ if and only if the \GacsKorner common
  information $K(X;Y)>0,$ i.e. if and only if $(X,Y)$ is
  \emph{decomposable}.
\end{enumerate}

The three key properties of maximal correlation that are useful for
the non-interactive simulation problem are as follows:

\vspace{0.1in}
\begin{itemize}
\item (data processing inequality) For any functions $\phi, \psi,$
  $\rho_m(X;Y)\geq \rho_m(\phi(X),\psi(Y)).$ 
\item (tensorization) If $(X_1,Y_1), (X_2,Y_2)$ are independent, then
  $\rho_m(X_1,X_2;Y_1,Y_2)=\max\{\rho_m(X_1;Y_1),\rho_m(X_2;Y_2)\}$
  \cite[Thm.~1]{Witsenhausen75}. 
\item (lower semi-continuity) (Recall that if $\mathcal{U}$ is a
  metric space, $u$ is a point in $\mathcal{U}$ and
  $f : \mathcal{U}\mapsto\mathbb{R}$ is a real-valued function, then
  we say $f$ is lower semi-continuous at $u$ if $u_n\to u$ implies
  $\liminf_n f(u_n) \geq f(u).$) If the space of probability
  distributions on $\mathcal{X}\times\mathcal{Y}$ is endowed with the
  total variation distance metric, then $\rho_m(X;Y)$ is a lower
  semi-continuous function of the joint distribution $P(x,y).$ [An
  example will be provided to show that $\rho_m$ is not a continuous
  function of the joint distribution.]
\end{itemize}
\vspace{0.1in}

To keep the paper self-contained, proofs of these properties are
sketched in Appendix~\ref{subsec:rho_m_properties}. Now, using the
above three properties, maximal correlation can be used to prove
impossibility results for the non-interactive simulation problem.

\vspace{0.1in}
\begin{adjustwidth}{0.6in}{0.6in}
\begin{observation}\label{observation:mc-machine}
  Non-interactive simulation of $(U,V)\sim Q(u,v)$ using $(X,Y)\sim
  P(x,y)$ is possible only if $\rho_m(X;Y)\geq \rho_m(U;V).$
\end{observation}
\end{adjustwidth}
\vspace{0.1in}

\begin{proof}Suppose non-interactive simulation of $(U,V)\sim Q(u,v)$
  using $(X,Y)\sim P(x,y)$ is possible. This means, there exists a
  sequence of integers $(k_n: n\geq 1),$ a sequence of finite
  alphabets $\mathcal{R}_n,$ and a sequence of functions
  $f_n:\mathcal{X}^{k_n}\times \mathcal{R}_n\mapsto\mathcal{U},$
  $g_n:\mathcal{Y}^{k_n}\times \mathcal{R}_n\mapsto\mathcal{V},$ such
  that if $\{X_i,Y_i\}_{i=1}^{k_n}$ are drawn i.i.d. $P(x,y)$ and
  $M_X,M_Y$ are uniformly distributed in $\mathcal{R}_n,$ with
  $\{X_i,Y_i\}_{i=1}^{k_n}, M_X, M_Y$ mutually independent, and
  $U_n=f_n(X^{k_n},M_X), V_n=g_n(Y^{k_n},M_Y),$ then
  $d_{\mathrm{TV}}((U_n,V_n);(U,V))\to 0$ as $n\to\infty.$ We
  therefore, have
\begin{align}
\rho_m(U_n;V_n) & \leq \rho_m(X^{k_n}, M_X;Y^{k_n}, M_Y) \mbox{ (Data
  Processing Inequality)} \\
& = \max\{\rho_m(X_1;Y_1), \rho_m(X_2;Y_2), \ldots,
\rho_m(X_{k_n},Y_{k_n}), \rho_m(M_X;M_Y)\} \mbox{ (Tensorization)} \\
& = \max\{\rho_m(X_1;Y_1), 0\} \\
& = \rho_m(X;Y) 
\end{align}
By lower semi-continuity of $\rho_m$,
$d_{\mathrm{TV}}((U_n,V_n);(U,V))\to 0$ implies
$$\rho_m(U;V)\leq \liminf_{n\to\infty}\rho_m(U_n;V_n)\leq \rho_m(X;Y).$$
\end{proof}

\subsection{Hypercontractivity}
\label{subsec:hc}

\begin{definition}For any real-valued random variable $W$ with finite
  support, and any real number $p,$ define

  \begin{align}||W||_p := \begin{cases} \left(\mathbb{E}|W|^p\right)^{1/p}, &
    p\neq 0; \\
    \exp\left(\mathbb{E}\log |W|\right) & p=0,
  \end{cases}
  \end{align}
  with the understanding that for $p\leq 0,$ $||W||_p=0$ if
  $\prob{|W|=0}>0.$

\end{definition}

$||W||_p$ is continuous and non-decreasing in $p.$ If $W$ is not
almost surely a constant, then $||W||_p$ is strictly increasing for
$p\geq 0.$ If in addition, $\prob{|W|=0}=0,$ then $||W||_p$ is
strictly increasing for all $p.$

\begin{definition}
  For any real $p\neq 0,1,$ define its \emph{H\"{o}lder conjugate} $p^\prime$ by
  $\frac{1}{p}+\frac{1}{p^\prime} = 1.$ For $p=0,$ define $p^\prime=0.$
\end{definition}

Suppose $X,Y$ are real-valued random variables with finite support. We
write $X\geq 0$ if $\prob{X\geq 0}=1.$ The following are well-known
\cite{HardyLittlewoodPolya}:
\begin{itemize}
\item (Minkowski's inequality) For $p\geq 1,$ $||X+Y||_p\leq ||X||_p+||Y||_p.$
\item (Reverse Minkowski's inequality) For $p\leq 1$ and $X,Y\geq 0,$ $||X+Y||_p\geq ||X||_p+||Y||_p.$
\item (H\"{o}lder's inequality) For $p>1,$ $\mathbb{E}[XY]\leq ||X||_{p^\prime}||Y||_p.$
\item (Reverse H\"{o}lder's inequality) For $p<1$ and $X,Y\geq 0,$ $\mathbb{E}[XY]\geq ||X||_{p^\prime}||Y||_p.$
\end{itemize}

\begin{definition}
  For a pair of random variables $(X,Y)\sim P(x,y)$ on
  $\mathcal{X}\times\mathcal{Y},$ we say $(X,Y)$ is
  $(p,q)$-hypercontractive if
  \begin{itemize}
  \item $1\leq q\leq p,$ and 
    \begin{equation}\label{eqn:hc}||\mathbb{E}[g(Y)|X]||_p\leq
      ||g(Y)||_q\ \ \forall g:\mathcal{Y}\mapsto\mathbb{R}
      ;\end{equation}
    (If $h(Y)=|g(Y)|,$ then  $-\mathbb{E}[h(Y)|X]\leq
    \mathbb{E}[g(Y)|X]\leq \mathbb{E}[h(Y)|X]$ pointwise, thus we 
    may equivalently restrict $g$ to map to $\mathbb{R}_{\geq
      0}.$ If $W_n$ supported on at most $k$ values (for some fixed $k$) converges to $W$ in distribution, then
    $||W_n||_p\to||W||_p$ for any $p,$ so we may further equivalently
    restrict $g$ to map to $\mathbb{R}_{>
      0}$.)
  \item $1\geq q\geq p,$ and 
    \begin{equation}\label{eqn:rhc}||\mathbb{E}[g(Y)|X]||_p\geq
      ||g(Y)||_q\ \ \forall
      g:\mathcal{Y}\mapsto\mathbb{R}_{\geq 0}.\end{equation}
    (If $W_n$ supported on at most $k$ values (for some fixed $k$) converges to $W$ in distribution, then
    $||W_n||_p\to||W||_p$ for any $p,$ so we may equivalently
    restrict $g$ to map to $\mathbb{R}_{>
      0}$.)

  \end{itemize}

  Note that in the conventional definitions in \eqref{eqn:hc} and
  \eqref{eqn:rhc}, we have functions taking values in $\mathbb{R}$ and
  $\mathbb{R}_{\geq 0}$ respectively. As explained above, for
  \eqref{eqn:hc}, we may restrict to functions taking values in
  $\mathbb{R}_{\geq 0}.$ However, in \eqref{eqn:rhc}, the functions
  must take non-negative values. This is conventional and necessary in
  various``reverse'' inequalities such as the reverse Minkowski and
  reverse H\"{o}lder inequalities.

  Define the \emph{hypercontractivity ribbon} $\mathcal{R}(X;Y)$ as
  the set of pairs $(p,q)$ for which $(X,Y)$ is
  $(p,q)$-hypercontractive.
\end{definition}

It is easy to check that the inequalities \eqref{eqn:hc},
\eqref{eqn:rhc} always hold for $p=q.$ The conditional expectation
operator is thus always contractive when $p\geq 1,$ and reverse
contractive for positive-valued functions when $p\leq 1.$ For random
variables $(X,Y)$ with a specific distribution $P(x,y),$ the operator
may be hypercontractive (i.e. more than contractive) in this precise
sense. $\mathcal{R}(X;Y)$ is a region in $\mathbb{R}^2$ pinching to a
point at $(1,1)$ resembling a ribbon, explaining our choice of the
name (see Fig.~\ref{fig:ribbon}). Inequality \eqref{eqn:rhc} is also
referred to as \emph{reverse hypercontractivity} in the literature
\cite{Mossel11}.




\begin{figure}[h]
  \begin{center}
    \includegraphics[viewport = 17 13 1594 1421, width = 2in, height=!]{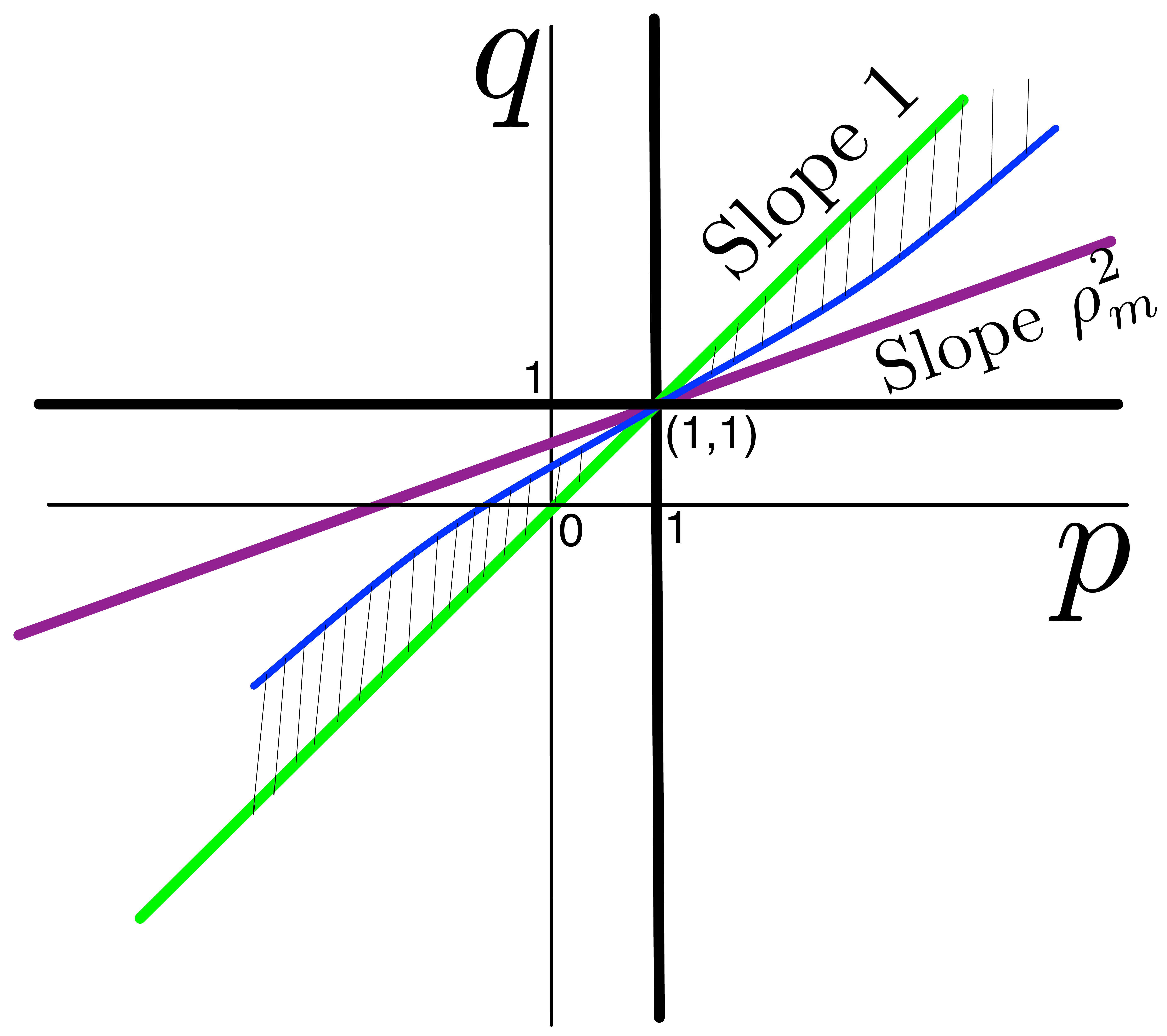}
    \caption{The hypercontractivity ribbon $\mathcal{R}(X;Y)$ is the
      shaded region. Also shown a straight line of slope
      $\rho_m^2:=\rho_m^2(X;Y)$ through $(1,1)$ (from Thm.~\ref{thm:main-inequality}).}
    \label{fig:ribbon}
  \end{center}
\end{figure}

\subsubsection{Interpretation of hypercontractivity as H\"{o}lder-contractivity}
  It is well-known \cite{Mossel11} that an equivalent definition of $\mathcal{R}(X;Y)$ can be given by
  observing how much the corresponding H\"{o}lder's and reverse
  H\"{o}lder's inequalities may be tightened:
  \begin{itemize}
  \item $(1,1)\in\mathcal{R}(X;Y);$
  \item For $1\leq q\leq p, 1<p$ we have $(p,q)\in \mathbb{R}(X;Y)$ iff
    \begin{align}
      \mathbb{E}f(X)g(Y)&\leq ||f(X)||_{p^\prime}||g(Y)||_q\ \ \  \forall f:\mathcal{X}\mapsto\mathbb{R}, 
      g:\mathcal{Y}\to\mathbb{R};\label{eqn:holder-tighten}
    \end{align}
  \item For $1\geq q\geq p, 1>p$ we have $(p,q)\in \mathbb{R}(X;Y)$ iff
    \begin{align}
      \mathbb{E}f(X)g(Y)&\geq ||f(X)||_{p^\prime}||g(Y)||_q\ \ \forall 
      f:\mathcal{X}\mapsto\mathbb{R}_{> 0}, 
      g:\mathcal{Y}\mapsto\mathbb{R}_{> 0};\label{eqn:reverse-holder-tighten}
    \end{align}
  \end{itemize}

  We will refer to inequalities \eqref{eqn:holder-tighten},
  \eqref{eqn:reverse-holder-tighten} as H\"{o}lder-contractive
  inequalities since they tighten H\"{o}lder's inequality (using the
  knowledge that $X$ and $Y$ are not `too correlated' in a suitable
  sense).

  To see the equivalence for $1\geq q \geq p, 1>p$ observe that if
  \eqref{eqn:rhc} holds for any strictly positive-valued function $g,$
  then for any fixed strictly positive-valued function $f,$ we have
\begin{align}
\mathbb{E}f(X)g(Y) &= \mathbb{E}\left[f(X)\mathbb{E}[g(Y)|X]\right] \\
& \geq ||f(X)||_{p^\prime} ||\mathbb{E}[g(Y)|X]||_{p}
\mbox{\ \ \ (Reverse H\"{o}lder's inequality and $\mathbb{E}[g(Y)|X]>0$)}\\
&\geq ||f(X)||_{p^\prime} ||g(Y)||_{q}.
\end{align}
Conversely, suppose \eqref{eqn:reverse-holder-tighten} holds for any
strictly positive-valued functions $f, g.$ First assume $p\neq 0.$ By fixing
$g$ and choosing $f(X) = \mathbb{E}[g(Y)|X]^{p-1},$ we get
\begin{align}
\mathbb{E}\left[\mathbb{E}[g(Y)|X]^p\right]& =
\mathbb{E}\left[\mathbb{E}[g(Y)|X]^{p-1}g(Y)\right] \\
& \geq ||\mathbb{E}[g(Y)|X]^{p-1}||_{p^\prime} ||g(Y)||_q \\
& =
\left(\mathbb{E}\left[\mathbb{E}[g(Y)|X]^p\right]\right)^{1-\frac{1}{p}}||g(Y)||_q.
\end{align}
Since $\mathbb{E}[g(Y)|X]>0,$ we obtain $||\mathbb{E}[g(Y)|X]||_p\geq ||g(Y)||_q.$

Now, consider the case $p=0.$ If \eqref{eqn:reverse-holder-tighten}
holds for any strictly positive-valued functions $f,g$ with
$p=p^\prime=0,$ then by monotonicity of $||\cdot||_r$ in $r,$ we also
have
    \begin{align}
      \mathbb{E}f(X)g(Y)&\geq ||f(X)||_{-\epsilon}||g(Y)||_q\ \ \forall 
      f:\mathcal{X}\mapsto\mathbb{R}_{> 0}, 
      g:\mathcal{Y}\mapsto\mathbb{R}_{> 0};\label{eqn:reverse-holder-epsilon}
    \end{align}
By our previous argument, this gives 
$||\mathbb{E}[g(Y)|X]||_{\frac{\epsilon}{1+\epsilon}}\geq ||g(Y)||_q.$
Since this holds for each $\epsilon>0,$ we get from continuity of
$||\cdot||_p$ in $p$ that $||\mathbb{E}[g(Y)|X]||_0\geq ||g(Y)||_q.$

The equivalence for the case $1\leq q\leq p, 1<p$ is similar. We only
need to note that for $(X,Y)$ to be $(p,q)$-hypercontractive with $1\leq q\leq p,$ it
suffices to have $||\mathbb{E}[g(Y)|X]||_p\leq ||g(Y)||_q$ hold only
for all strictly positive functions $g>0.$ The rest of the proof is identical.

\subsubsection{Duality between $\mathcal{R}(X;Y)$ and $\mathcal{R}(Y;X)$}

The equivalent description of $\mathcal{R}(X;Y)$ in
\eqref{eqn:holder-tighten}, \eqref{eqn:reverse-holder-tighten}
immediately gives the following duality between $\mathcal{R}(X;Y)$ and
$\mathcal{R}(Y;X)$:

\begin{align}
(p,q)\in\mathcal{R}(X;Y) \Leftrightarrow (q^\prime,
p^\prime)\in\mathcal{R}(Y;X), \ \ \ p,q\neq 1.
\end{align}

$\mathcal{R}(X;Y)$ is completely specified by its non-trivial boundary
$q^*_p(X;Y)$ defined for $p\neq 1$ as 
\begin{align}
  q^*_p(X;Y) & := \begin{cases} \inf\{q\geq 1: ||\mathbb{E}[g(Y)|X]||_p\leq
    ||g(Y)||_q \ \ \forall g:\mathcal{Y}\mapsto\mathbb{R}\} & p>1; \\
    \sup\{q\leq 1: ||\mathbb{E}[g(Y)|X]||_p\geq
    ||g(Y)||_q \ \ \forall g:\mathcal{Y}\mapsto\mathbb{R}_{>0}\} &
    p<1.
  \end{cases}
\end{align}

We will find it useful to define the `slope at $p$' by $s_p(X;Y) := \frac{q^*_p(X;Y)-1}{p-1}$
for $p\neq 1.$ 

The following properties may be easily shown.
\begin{enumerate}
\item $0\leq s_p(X;Y)\leq 1.$
\item $s_p(X;Y)=0$ if and only if $X$ is independent of $Y.$ [This is a
  consequence of Thm.~\ref{thm:main-inequality} and the
  corresponding property for $\rho_m(X;Y).$]
\end{enumerate}

One can show that for any $p\neq 1,$ $s_p(X;Y)$ satisfies the same three
key properties that maximal correlation satisfies (proofs of these
properties are sketched in Appendix~\ref{subsec:s_p_properties}).
\vspace{0.1in}
\begin{itemize}
\item (data processing inequality) For any functions $\phi, \psi,$
  $s_p(X;Y)\geq s_p(\phi(X),\psi(Y)).$
\item (tensorization) If $(X_1,Y_1), (X_2,Y_2)$ are independent, then
  $s_p(X_1,X_2;Y_1,Y_2)=\max\{s_p(X_1;Y_1),s_p(X_2;Y_2)\}$
  \cite{Witsenhausen75}.
\item (lower semi-continuity) If the space of probability
  distributions on $\mathcal{X}\times\mathcal{Y}$ is endowed with the
  total variation distance metric, then $s_p(X;Y)$ is a lower
  semi-continuous function of the joint distribution $P(x,y).$ [An
  example will be provided to show that $s_p$ is not a continuous
  function of the joint distribution.]
\end{itemize}
\vspace{0.1in}

Thus, we can use hypercontractivity to obtain impossibility results
for the non-interactive simulation problem.

\vspace{0.1in}
\begin{adjustwidth}{0.6in}{0.6in}
\begin{observation}\label{observation:hc-machine}
  Non-interactive simulation of $(U,V)\sim Q(u,v)$ using $(X,Y)\sim
  P(x,y)$ is possible only if $s_p(X;Y)\geq s_p(U;V)$ for each $p\neq
  1,$ in other words, only if
  $\mathcal{R}(X;Y)\subseteq\mathcal{R}(U;V).$
\end{observation}
\end{adjustwidth}
\vspace{0.1in}

\begin{example}
A classical result states that for $(X,Y)\sim\DSBS(\alpha),$

\begin{align}
\frac{q_p^*(X;Y)-1}{p-1} = s_p(X;Y) = (1-2\alpha)^2, \ \ p\neq 1.\label{eqn:DSBS-hc}
\end{align}

This was proved by Bonami \cite{Bonami70} and Beckner \cite[Lemma~1,
Appendix~Sec.~2]{Beckner} for $p>1$ and by Borell
\cite[Thm~3.2]{Borell} for $p<1.$
\end{example}

\subsection{Proving impossibility results for non-interactive simulation
  using the hypercontractivity ribbon $\mathcal{R}(X;Y)$}
\label{subsec:hc-ribbon}

In this subsection, we state explicitly a simple observation that is
well-known.  Suppose non-interactive simulation of $(U,V)\sim Q(u,v)$
using $(X,Y)\sim P(x,y)$ is possible. This means, there exists a
sequence of integers $(k_n: n\geq 1),$ a sequence of finite alphabets
$\mathcal{R}_n,$ and a sequence of functions
$f_n:\mathcal{X}^{k_n}\times \mathcal{R}_n\mapsto\mathcal{U},$
$g_n:\mathcal{Y}^{k_n}\times \mathcal{R}_n\mapsto\mathcal{V},$ such
that if $\{X_i,Y_i\}_{i=1}^{k_n}$ are drawn i.i.d. $P(x,y)$ and
$M_X,M_Y$ are uniformly distributed in $\mathcal{R}_n,$ with
$\{X_i,Y_i\}_{i=1}^{k_n}, M_X, M_Y$ mutually independent, and
$U_n=f_n(X^{k_n},M_X), V_n=g_n(Y^{k_n},M_Y),$ then
$d_{\mathrm{TV}}((U_n,V_n);(U,V))\to 0$ as $n\to\infty.$ Let
$(U_n,V_n)\sim Q_n(u,v).$

A traditional approach to prove impossibility results for
non-interactive simulation is as follows. Fix $n.$ Suppose $(X,Y)$ is
$(p,q)$-hypercontractive with $1\leq q\leq p.$ Then, by tensorization
$((X^{k_n},M_X),(Y^{k_n},M_Y))$ is $(p,q)$-hypercontractive.

Consider the functions $\phi_n,\psi_n$ defined as:
\begin{align}\phi_n(x^{k_n},m_x) = \sum_{u\in \mathcal{U}}\lambda_u
  \indicator{f_n(x^{k_n},m_x)=u},
\label{eqn:f(x^n)}
\end{align}
\begin{align}
  \psi_n(y^{k_n},m_y) = \sum_{v\in \mathcal{V}}\mu_v
  \indicator{g_n(y^{k_n},m_y)=v}.
\label{eqn:g(y^n)}
\end{align}
By using
\eqref{eqn:holder-tighten}, we get
\begin{align}
\mathbb{E}\phi_n(X^{k_n},M_X)\psi(Y^{k_n},M_Y)\leq ||\phi(X^{k_n},M_X)||_{p^\prime}||\psi(Y^{k_n},M_Y)||_q,
\end{align}
which is 

\begin{align}
  \sum_{u\in \mathcal{U}} \sum_{v\in \mathcal{V}}\lambda_u\mu_v Q_n(u,v)
  & \leq \left(\sum_{u\in \mathcal{U}} \lambda_u^{p^\prime}
    Q_n(u)\right)^{1/p^\prime}\cdot \left(\sum_{v\in \mathcal{V}}
    \mu_v^{q} Q_n(v)\right)^{1/q}.
\end{align}
By letting $n\to\infty,$ we get 

\begin{align}
  \sum_{u\in \mathcal{U}} \sum_{v\in \mathcal{V}}\lambda_u\mu_v Q(u,v)
  & \leq \left(\sum_{u\in \mathcal{U}} \lambda_u^{p^\prime}
    Q(u)\right)^{1/p^\prime}\cdot \left(\sum_{v\in \mathcal{V}}
    \mu_v^{q} Q(v)\right)^{1/q}.\label{eqn:two-function-hc-bound}
\end{align}

For any fixed $\lambda_u,\mu_v,$ we find that non-interactive
simulation of $(U,V)\sim Q(u,v)$ from $(X,Y)\sim P(x,y)$ is possible
only if $Q$ satisfies the inequality \eqref{eqn:two-function-hc-bound}.

Similarly, if $(X,Y)$ is $(p,q)$-hypercontractive with $1\geq q\geq p$ then,
for any fixed $\lambda_u,\mu_v>0,$ non-interactive simulation of
$(U,V)\sim Q(u,v)$ from $(X,Y)\sim P(x,y)$ is possible only if $Q$
satisfies the following inequality:
\begin{align}
\sum_{u\in \mathcal{U}} \sum_{v\in \mathcal{V}}\lambda_u\mu_v
  Q(u,v)
  & \geq \left(\sum_{u\in \mathcal{U}} \lambda_u^{p^\prime}
    Q(u)\right)^{1/p^\prime}\cdot \left(\sum_{v\in \mathcal{V}}
    \mu_v^{q} Q(v)\right)^{1/q}.\label{eqn:two-function-rhc-bound}
\end{align}

Indeed, \eqref{eqn:basic-reverse-hc} is a version of
\eqref{eqn:two-function-rhc-bound}. Let $(X,Y)\sim\DSBS(\alpha).$
Then, $(X,Y)$ is
$(-\frac{2\alpha}{1-2\alpha},2\alpha)$-hypercontractive from
\eqref{eqn:DSBS-hc}. Choosing $\lambda_0=\mu_0=1,
\lambda_1=\mu_1=\epsilon$ with $\epsilon\to 0,$ we obtain
\eqref{eqn:basic-reverse-hc} where $\mathcal{U} = \mathcal{V} =
\{0,1\}.$

The inclusion $\mathcal{R}(X;Y)\subseteq\mathcal{R}(U;V)$ implies the
collection of inequalities \eqref{eqn:two-function-hc-bound} for any
choice of real
$\{\lambda_u\}_{u\in\mathcal{U}},\{\mu_v\}_{v\in\mathcal{V}}$ and the
collection of inequalities \eqref{eqn:two-function-rhc-bound} for any
choice of positive valued
$\{\lambda_u\}_{u\in\mathcal{U}},\{\mu_v\}_{v\in\mathcal{V}}.$ One can
also easily show that the reverse implication from the collection of
inequalities \eqref{eqn:two-function-hc-bound},
\eqref{eqn:two-function-rhc-bound} to
$\mathcal{R}(X;Y)\subseteq\mathcal{R}(U;V)$ holds (using the
equivalent interpretation of hypercontractivity as
H\"{o}lder-contractivity).

Thus, $\mathcal{R}(X;Y)\subseteq\mathcal{R}(U;V)$ is powerful enough
to subsume the application of all possible instantiations of
$\lambda_u,\mu_v$ in the corresponding H\"{o}lder-contractive
inequalities.

The reader should note the importance of the above observation in the
context of thinking abstractly about the hypercontractivity ribbon and
its usefulness when invoking an automated computer search for proving
an impossibility of non-interactive simulation result. If non-interactive
simulation of $(U,V)$ using $(X,Y)$ is possible, then any
H\"{o}lder-contractive inequality satisfied by $(X,Y)$ will also be
satisfied by $(U,V).$ Therefore, if any such inequality satisfied by
all functions of $X$ and $Y$ is violated by some pair of functions of
$U$ and $V,$ then we can conlude non-simulability, i.e. that
simulation of $(U,V)$ using $(X,Y)$ is impossible. However, violation
of any such H\"{o}lder-contractive inequality implies failure of the
inclusion $\mathcal{R}(X;Y)\subseteq \mathcal{R}(U;V),$ so one can get
the same conclusion from the result that failure of the inclusion
$\mathcal{R}(X;Y)\subseteq \mathcal{R}(U;V)$ implies
non-simulability. Further, it is easier to show failure of inclusion
of the hypercontractivity ribbons than it is to show violation of any
specific such H\"{o}lder-contractive inequality, simply because
violation of any H\"{o}lder-contractive inequality implies failure of
inclusion of the hypercontractivity ribbons but failure of inclusion
of the hypercontractivity ribbons just implies that \emph{some}
H\"{o}lder-contractive inequality is violated. Thus, if one wishes to
show non-simulability using a computer search, it suffices to compute
the non-trivial boundaries of the two hypercontractivity ribbons
$q_p^*(X;Y)$ and $q_p^*(U;V)$ (and the corresponding $s_p(X;Y)$ and
$s_p(U;V)$) and find that $s_p(X;Y)<s_p(U;V)$ for some $p\neq 1$
without ever having to prove for some specific H\"{o}lder-contractive
inequality that it is the one being violated. 

To the best of our knowledge, there is no algorithm better than a
brute force search following suitable discretization to compute the
hypercontractivity ribbons. However, the observation above simplifies
the approach of proving an impossibility result using instantiations
of $\lambda_u$ and $\mu_v.$

\section{Main Results}
\label{sec:main-results}

In this section, we state and prove our main results. 

\subsection{Connection between maximal correlation and the
  hypercontractivity ribbon}
Our first result is a
geometric connection between maximal correlation and the
hypercontractivity ribbon. 

\begin{theorem}\label{thm:main-inequality}
If $(X,Y)$ is $(p,q)$-hypercontractive and $p\neq 1,$ then 
\begin{equation}\label{eqn:main-inequality}
\rho_m^2(X;Y)\leq \frac{q-1}{p-1}.
\end{equation}
\end{theorem}

\begin{remark}
  For the case $p>1,$ Thm.~\ref{thm:main-inequality} is obtained in
  \cite{AhlswedeGacs76}. In the current form of the statement of
  Thm.~\ref{thm:main-inequality}, the maximal correlation is
  afforded a geometric meaning, namely its square is the slope of a
  straight line bound constraining the hypercontractivity ribbon (see
  Fig~\ref{fig:ribbon}). For $(X,Y)\sim\DSBS(\alpha),$ we have from
  \eqref{eqn:DSBS-mc} and \eqref{eqn:DSBS-hc} that the
  hypercontractivity ribbon $\mathcal{R}(X;Y)$ is precisely the wedge
  obtained by the straight lines $p=q,$ and the straight line
  corresponding to the maximal correlation bound
  $\frac{q-1}{p-1}=\rho_m^2(X;Y).$
\end{remark}

\begin{proof}[Proof of Theorem~\ref{thm:main-inequality}]
  The proof uses a perturbative argument. Let $(X,Y)\sim P(x,y).$ The
  claim is obvious when either $X$ or $Y$ is a constant almost
  surely. So, assume this is not the case and fix functions
  $\phi:\mathcal{X}\mapsto\mathbb{R},
  \psi:\mathcal{Y}\mapsto\mathbb{R}$ such that
\begin{equation}\label{eqn:conditions}
\mathbb{E}\phi(X)=\mathbb{E}\psi(Y)=0,\ 
\mathbb{E}\phi(X)^2=\mathbb{E}\psi(Y)^2=1.
\end{equation}
Fix $r>0.$ Define $f:\mathcal{X}\mapsto\mathbb{R}_{>0},
g:\mathcal{Y}\mapsto\mathbb{R}_{>0}$ by $f(x)=1+\frac{\sigma}{r}
\phi(x), g(y)=1+\sigma r\psi(y).$ Note that for sufficiently small
$\sigma,$ the functions $f, g$ do take only positive values. Fix
$(p,q)\in\mathbb{R}_{X;Y}$ with $p<1.$ We also assume $p\neq 0$ using
the standard limit argument to deal with the case $p=0.$ Using
\eqref{eqn:reverse-holder-tighten} with the functions $f,g$ we just
defined, we have
\begin{align}
  \mathbb{E}[(1+\frac{\sigma}{r}\phi(X))(1+\sigma r\psi(Y))] \geq &
  \left(\mathbb{E}[(1+\frac{\sigma}{r}\phi(X))^{p^\prime}]\right)^{1/p^\prime}\cdot
  \left(\mathbb{E}[(1+\sigma
    r\psi(Y))^{q}]\right)^{1/q}. \label{eqn:perturbation}
\end{align}
For $Z$ satisfying $\mathbb{E}Z=0, \mathbb{E}Z^2=1,$
\begin{align*}
  \left(\mathbb{E}[(1+ a Z)^l]\right)^{1/l} & = \left(1+l\cdot
    a\mathbb{E}
    Z+\frac{l(l-1)}{2}\cdot a^2\mathbb{E}Z^2+O(a^3)\right)^{1/l} \\
  & =1+\frac{l-1}{2}a^2+O(a^3).
\end{align*}
Using this in \eqref{eqn:perturbation}, we get
\begin{align*}
1 + \sigma^2\mathbb{E}[\phi(X)\psi(Y)] \geq \left(1+ \frac{p^\prime-1}{2r^2}\sigma^2+O(\sigma^3)\right)\left(1+\frac{(q-1)r^2}{2}\sigma^2+O(\sigma^3)\right)~.
\end{align*}
Comparing the coefficient of $\sigma^2$ on both sides, we get
\begin{align*}
  \mathbb{E}\phi(X)\psi(Y)\geq
\frac{p^\prime-1}{2r^2}+\frac{(q-1)r^2}{2}.
\end{align*}
Noting that $p^\prime-1,q-1<0$ and taking the supremum over all $r>0,$
we get 
\begin{align}
  \mathbb{E}\phi(X)\psi(Y)\geq -\sqrt{\frac{q-1}{p-1}} \mbox{\ \ \ \
    or\ \ \ \ }   -\mathbb{E}\phi(X)\psi(Y)\leq \sqrt{\frac{q-1}{p-1}}.
\end{align}

Taking the supremum over all $-\phi$ and $\psi$ satisfying
\eqref{eqn:conditions}, we get
$$\rho_m(X;Y)\leq \sqrt{\frac{q-1}{p-1}}.$$
We can similarly prove the inequality in the case when $p>1.$ This
completes the proof.
\end{proof}

\vspace{0.1in}

The main implication of Thm.~\ref{thm:main-inequality} for the
problem of non-interactive simulation is the following corollary,
which gives a necessary and sufficient condition on the source
distribution $P(x,y)$ for which
Observation~\ref{observation:hc-machine} will prove impossibility
results that are at least as strong as
Observation~\ref{observation:mc-machine}. This condition is satisfied
for example, when $P(x,y)$ is a $\DSBS(\epsilon)$ distribution.

\vspace{0.1in}

\begin{corollary}\label{cor:mc-hc-comparison}
Fix a distribution $(X,Y)\sim P(x,y).$ Then the following are equivalent:
\begin{itemize}
\item[(a)] For all $(U,V)\sim Q(u,v),$ $\mathcal{R}(X;Y)\subseteq \mathcal{R}(U;V) \implies \rho_m(X;Y)\geq
\rho_m(U;V).$
\item[(b)] \begin{equation}\label{eqn-main-inequality-equality}
\rho_m(X;Y)=\inf_{(p,q)\in\mathcal{R}(X;Y),p\neq 1}
\sqrt{\frac{q-1}{p-1}}.\end{equation}
\end{itemize}

\end{corollary}

\vspace{0.1in}

\begin{proof}[Proof of Corollary~\ref{cor:mc-hc-comparison}]

  (b) $\implies $ (a): Assume (b) holds for $P(x,y).$ If
  $\mathcal{R}(X;Y)\subseteq\mathcal{R}(U;V),$ then
  $\inf_{(p,q)\in\mathcal{R}(X;Y),p\neq 1} \sqrt{\frac{q-1}{p-1}}\geq
  \inf_{(p,q)\in\mathcal{R}(U;V),p\neq 1} \sqrt{\frac{q-1}{p-1}}.$
  Now, by hypothesis, $\inf_{(p,q)\in\mathcal{R}(X;Y),p\neq 1}
  \sqrt{\frac{q-1}{p-1}} = \rho_m(X;Y)$ and from
  Thm.~\ref{thm:main-inequality}, we have
  $\inf_{(p,q)\in\mathcal{R}(U;V),p\neq 1} \sqrt{\frac{q-1}{p-1}}\geq
  \rho_m(U;V).$

  $\sim$(b) $\implies $ $\sim$(a): Suppose that for $(X,Y)\sim P(x,y),$ we have
  for some $\delta\neq 0,$
  $$\rho_m(X;Y)= \inf_{(p,q)\in\mathcal{R}(X;Y),p\neq 1}
  \sqrt{\frac{q-1}{p-1}}-\delta.$$ By
  Theorem~\ref{thm:main-inequality}, $\delta>0.$ From \eqref{eqn:DSBS-hc}, we know
  that if $(U,V)\sim \DSBS(\epsilon),$ then for any $p\neq 1,$
  $$\frac{q^*_p(U;V)-1}{p-1}=(1-2\epsilon)^2=\rho_m(U;V)^2.$$
  Choosing $\epsilon$ so that
  $\rho_m(U;V)=1-2\epsilon=\inf_{(p,q)\in\mathcal{R}(X;Y),p\neq 1}
  \sqrt{\frac{q-1}{p-1}},$ we have $\rho_m(X;Y)<\rho_m(U;V)$ and
  $\mathcal{R}(X;Y)\subseteq\mathcal{R}(U;V).$

\end{proof}

\subsection{Limiting chordal slope of the hypercontractivity ribbon}

Our second result proves the existence of $\lim_{p\to 1}s_p(X;Y)$ and
provides a characterization of the limit in terms of a strong data
processing constant for relative entropies that was studied first in
\cite{AhlswedeGacs76}.

\begin{definition}
  Let $D(\mu(z)||\nu(z)) = \sum_z\mu(z)\log\frac{\mu(z)}{\nu(z)}$
  denote the relative entropy of $\mu$ with respect to $\nu.$ Consider
  finite sets $\mathcal{X}$ and $\mathcal{Y},$ and let $P(x,y)$ be a
  joint distribution over the product set
  $\mathcal{X}\times\mathcal{Y}.$ Let $R_X(x)$ be an arbitrary
  probability distribution on $\mathcal{X}.$ Let $R_Y(y)$ be the
  probability distribution on $\mathcal{Y}$ whose probability mass at
  $y$ is $\sum_{x\in\mathcal{X}} \frac{P(x,y)}{P_X(x)}R_X(x).$ If
  $(X,Y)\sim P_X(x,y),$ then define the strong data processing constant
  for relative entropies corresponding to $(X,Y)$ as
  $$s^*(X;Y) := \sup\frac{D(R_Y(y)||P_Y(y))}{D(R_X(x)||P_X(x))}\ ,$$
  where the supremum is taken over all $R_X(x)$
  satisfying $R_X(x)\not\equiv P_X(x)$ and $R_X(x)<<P_X(x).$
\end{definition}

\begin{remark}
  In a recent work \cite{AGKN_14}, it is shown that $s^*$ is also the
  tightest constant for data processing inequalities involving mutual
  information in Markov chains:
  $$s^*(X;Y) = \sup_{U: U-X-Y} \frac{I(U;Y)}{I(U;X)}~.$$
\end{remark}

Our result can be stated as follows. 

\vspace{0.1in}
\begin{theorem}\label{thm:limit_p=1}
  \begin{equation}\label{eqn:limit_p=1}\lim_{p\to 1} s_p(X;Y) =
    \lim_{p\to 1} \frac{q^*_p(X;Y)-1}{p-1} = s^*(Y;X).\end{equation}
\end{theorem}

\vspace{0.1in} The proof of Thm.~\ref{thm:limit_p=1} follows from a
natural Taylor series calculation, and can be found in
Appendix~\ref{subsec:s*_properties}. The following corollary shows
that
$\lim_{p\to\infty} s_p(X;Y) =\lim_{p\to-\infty} s_p(X;Y) = s^*(X;Y).$
The former was established in \cite{AhlswedeGacs76} while the latter
result is new. We believe that using
Theorems~\ref{thm:main-inequality} and~\ref{thm:limit_p=1}, we acquire
a more intuitive proof of the result
$\lim_{p\to\infty} s_p(X;Y) = s^*(X;Y)$ that was obtained in
\cite{AhlswedeGacs76}, while also showing the reverse hypercontractive
case: $\lim_{p\to-\infty} s_p(X;Y) = s^*(X;Y)$

\vspace{0.1in}
\begin{corollary}\label{cor:limit_p=infty}
\begin{equation}\label{eqn:limit_p=infty}\lim_{p\to \infty}
  \frac{q^*_p(X;Y)-1}{p-1} = \lim_{p\to-\infty}
  \frac{q^*_p(X;Y)-1}{p-1} = s^*(X;Y).\end{equation}
\end{corollary}
\vspace{0.1in}

The proof of Corollary~\ref{cor:limit_p=infty} is in
Appendix~\ref{subsec:s*_properties}.
Corollary~\ref{cor:sufficient_condition}, which follows immediately
from Corollary~\ref{cor:mc-hc-comparison}, Thm.~\ref{thm:limit_p=1}
and Corollary~\ref{cor:limit_p=infty} provides a sufficient condition
for \eqref{eqn-main-inequality-equality} to hold.  \vspace{0.1in}
\begin{corollary}\label{cor:sufficient_condition}
If $\rho_m(X;Y)=\min\{\sqrt{s^*(X;Y)}, \sqrt{s^*(Y;X)}\},$ then for
any $(U,V)\sim Q(u,v),$ we have 
$$\mathcal{R}(X;Y)\subseteq \mathcal{R}(U;V) \implies \rho_m(X;Y)\geq
\rho_m(U;V).$$
\end{corollary}
\vspace{0.1in} Note that from \eqref{eqn:DSBS-mc}, \eqref{eqn:DSBS-hc}
and Thm.~\ref{thm:limit_p=1}, DSBS sources always satisfy the
condition in Corollary~\ref{cor:sufficient_condition}. One can also
show that the condition holds for source distributions corresponding
to the input-output pair resulting from a uniformly distributed input
into a binary input symmetric output channel. The above ideas suggest
that for a recent conjecture regarding Boolean functions
\cite{KumarCourtade13}, hypercontractivity is going to be a more
useful tool than maximal correlation. Indeed, evidence for this can be
found in \cite{AGKN_Allerton13}, where usage of $s^*$ helps in an
automated proof of an inequality that cannot be proved using maximal
correlation.

\begin{example}

  Suppose we choose $P(x,y)$ to be $\DSBS(\epsilon),$ and $Q(u,v)$
  specified by $Q(U=1) = s,$ $Q(V=1|U=0) = c, Q(V=0|U=1)=d.$ For
  certain values of $s,c,d,$ non-interactive simulation is possible
  and for others, it is impossible. For fixed values of $s,$ this is
  shown graphically in Fig.~\ref{fig:rho_m_s_p}.

\begin{figure}[h]
  \begin{center}
    \subfigure[$s=0.3,\epsilon=0.2$]{\includegraphics[width=1.8in,
      height=!]{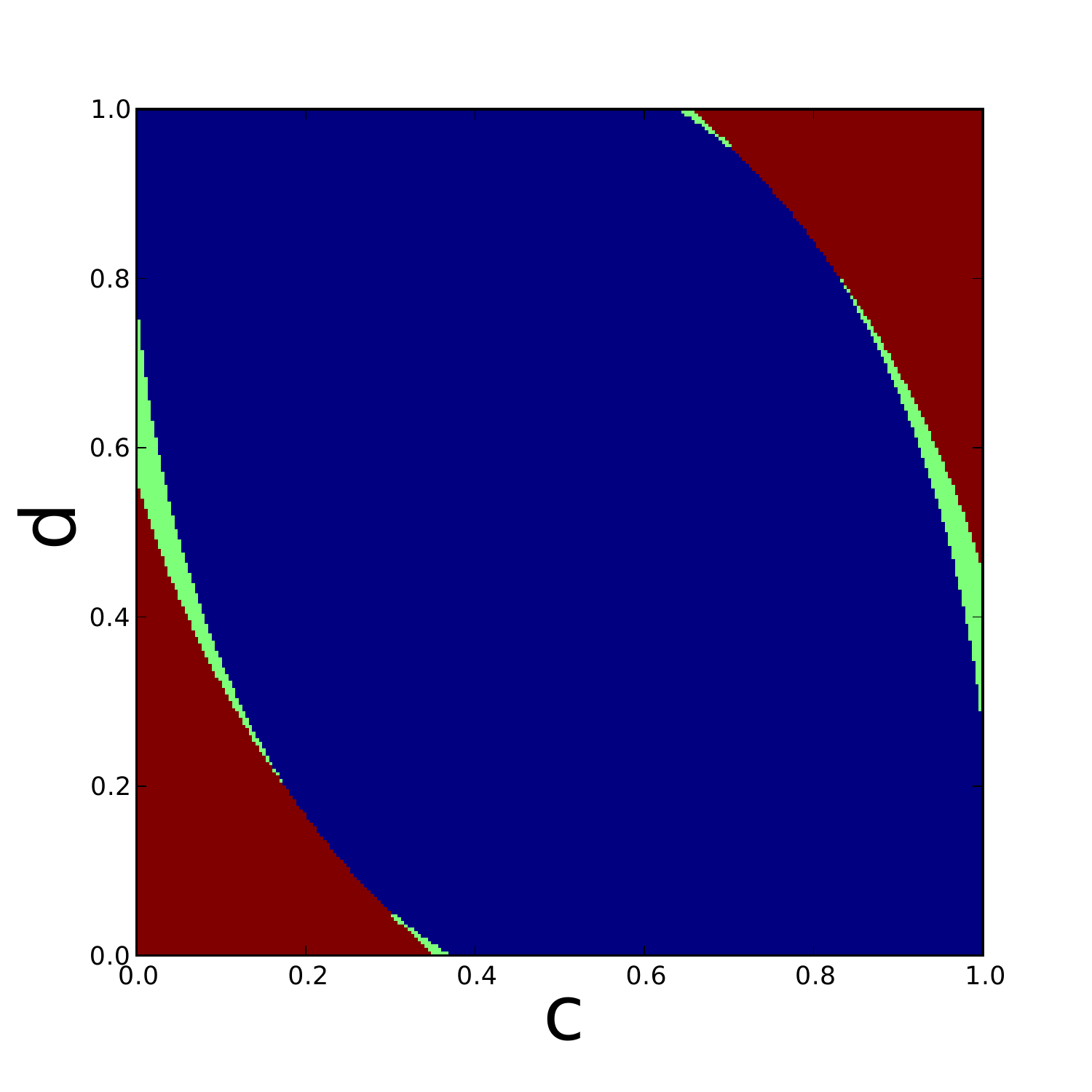}}\hspace{0.1in}
    \subfigure[$s=0.3,\epsilon=0.4$]{\includegraphics[width=1.8in,
      height=!]{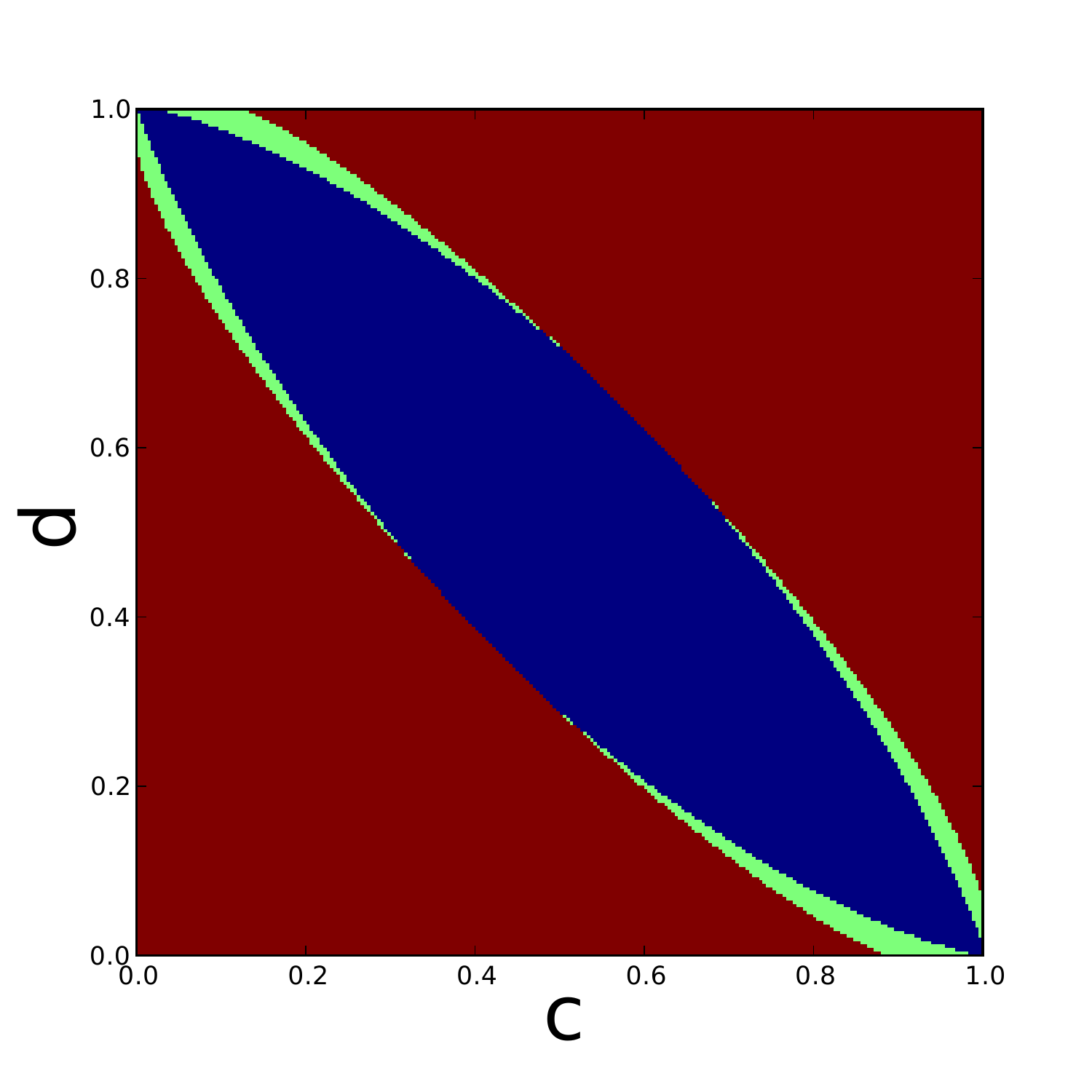}}
    \subfigure[$s=0.5,\epsilon=0.2$]{\includegraphics[width=1.8in,
      height=!]{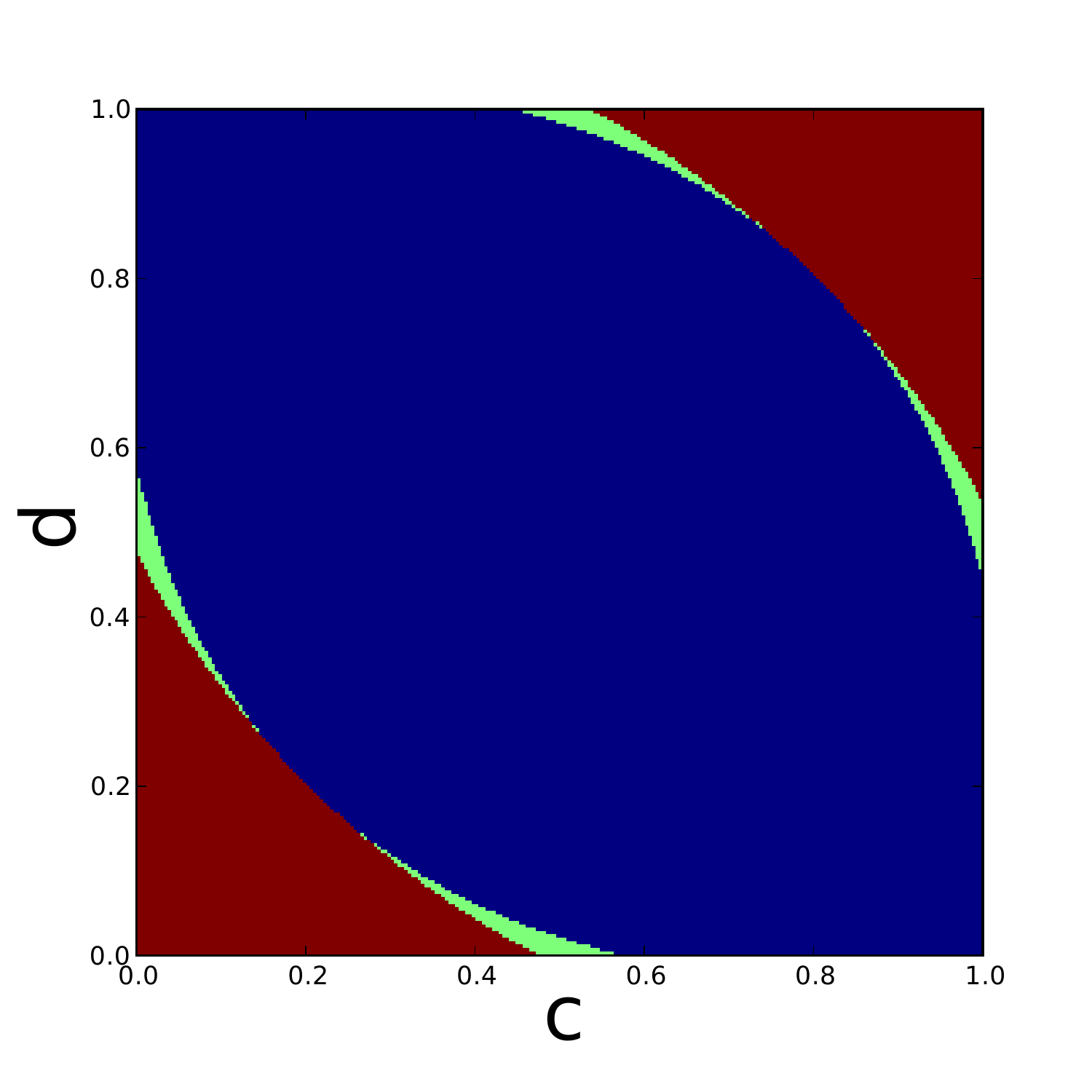}}
    \caption{Suppose the source distribution $P(x,y)$ is
      $\DSBS(\epsilon),$ and the target distribution is $Q(u,v)$
      specified by $Q(U=1) = s,$ $Q(V=1|U=0) = c, Q(V=0|U=1)=d.$ The
      plots above show restrictions on the space of distributions
      $(s,c,d)$ that can be simulated.  The X co-ordinate represents
      $c$ and the Y co-ordinate represents $d.$ In each plot, we fix
      $s,$ $\epsilon$ as specifed and $p=1.5.$ The blue region
      indicates $\rho_m^2\leq s_p\leq (1-2\epsilon)^2,$ the green
      region indicates $\rho_m^2\leq (1-2\epsilon)^2<s_p$ and finally,
      the red region indicates $(1-2\epsilon)^2<\rho_m^2\leq s_p.$
      Thus, with $p=1.5,$ the red region is ruled out as impossible by
      $\rho_m$ and $s_p$, the green region is ruled out by $s_p,$ and
      the blue region is ruled out by neither $\rho_m$ nor by $s_p.$
      Note that this does not mean all points in the blue region can
      be simulated by suitable choice of functions, only that our
      tools (using this particular choice of $p$) fail to prove
      impossibility for those points. Note that along the $c=d$ line,
      $(U,V)$ is a DSBS source as well, so both maximal correlation
      and hypercontractivity (for any $p$) give an impossibility
      result if and only if $c<\epsilon$ or $c>1-\epsilon$ in
      accordance with Sec.~\ref{subsec:example1}.}
    \label{fig:rho_m_s_p}
  \end{center}
\end{figure}

\end{example}


\section{Non-interactive simulation with $k\geq 3$ agents}
\label{sec:triple}

The non-interactive simulation problem we have considered can be
naturally extended to $k$-agents.

\begin{definition}
Let $\mathcal{X}_i, \mathcal{U}_i$ denote finite sets for
$i=1,2,\ldots, k.$ Given a \textit{source distribution}
$P(x_1,x_2,\ldots, x_k)$ over $\Pi_{i=1}^k\mathcal{X}_i$ and a
\textit{target distribution} $Q(u_1,u_2,\ldots, u_k)$ over
$\Pi_{i=1}^k\mathcal{U}_i,$ we say that \emph{non-interactive
  simulation} of $Q(u_1,u_2,\ldots, u_k)$ using $P(x_1,x_2,\ldots,
x_k)$ is possible if for any $\epsilon>0,$ there exists a positive
integer $n,$ a finite set $\mathcal{R}$
and functions $f_i:\mathcal{X}_i^n\times\mathcal{R}\mapsto
\mathcal{U}_i$ for $i=1,2,\ldots,k$ such that
$$d_{\mathrm{TV}}\left((f_1(X_1^n,M_1),f_2(X_2^n,M_2),\ldots,
  f_k(X_k^n,M_k));(U_1,U_2,\ldots, U_k)\right)\leq \epsilon$$ where
$\{(X_{1,j},X_{2,j},\ldots, X_{k,j})\}_{j=1}^n$ is a sequence of i.i.d. samples
drawn from $P(x_1,x_2,\ldots, x_k),$ $M_1,M_2,\ldots, M_k$ are
uniformly distributed in $\mathcal{R},$ mutually independent of each
other and of the samples drawn from the source, $(U_1,U_2,\ldots,
U_k)$ is drawn from $Q(u_1,u_2,\ldots, u_k),$ and
$d_{\mathrm{TV}}(\cdot\,;\cdot)$ is the total variation distance.
\end{definition}

In this section, we make simple observations about how
hypercontractivity and maximal correlation may be used to prove
impossibility results for this non-interactive simulation problem with
$k$ agents. For any set $A\subseteq\{1,2,\ldots, k\},$ let us use the
notation $X_A:=(X_i: i\in A), U_A:=(U_i:i\in A).$

Recall that for the case of two random variables $(X,Y)$ and $1\leq q<p,$ we
have $(p,q)\in\mathcal{R}(X;Y)$ if either of the two following
equivalent conditions hold:
\vspace{0.1in}
\begin{itemize}
\item $||\mathbb{E}[g(Y)|X]||_p\leq ||g(Y)||_q\ \ \forall g:\mathcal{Y}\mapsto\mathbb{R};$
\item $\mathbb{E}f(X)g(Y)\leq ||f(X)||_{p^\prime}||g(Y)||_q\ \ \forall
  f:\mathcal{X}\mapsto\mathbb{R}, \forall g:\mathcal{Y}\mapsto\mathbb{R}.$
\end{itemize}
\vspace{0.1in}
Similarly, for $1\geq q>p,$ we have $(p,q)\in\mathcal{R}(X;Y)$ if
either of the two following equivalent conditions hold:
\vspace{0.1in}
\begin{itemize}
\item $||\mathbb{E}[g(Y)|X]||_p\geq ||g(Y)||_q\ \ \forall
  g:\mathcal{Y}\mapsto\mathbb{R}_{>0}$;
\item $\mathbb{E}f(X)g(Y)\geq ||f(X)||_{p^\prime}||g(Y)||_q\ \ \forall
  f:\mathcal{X}\mapsto\mathbb{R}_{>0}, \forall g:\mathcal{Y}\mapsto\mathbb{R}_{>0}.$
\end{itemize}
\vspace{0.1in} We can define a H\"{o}lder-contraction region
$\mathcal{H}(X;Y)$ by observing how much H\"{o}lder's inequality and
the reverse H\"{o}lder's inequality may be tightened. Define
$(p_1,p_2)\in\mathcal{H}(X;Y)$ if \vspace{0.1in}
\begin{itemize}
\item $p_1,p_2\geq 1,$ and $\forall f:\mathcal{X}\mapsto\mathbb{R},
  \forall g:\mathcal{Y}\mapsto\mathbb{R},$ we have
  $\mathbb{E}f(X)g(Y)\leq ||f(X)||_{p_1}||g(Y)||_{p_2};$
\item $p_1,p_2\leq 1,$ and $\forall
  f:\mathcal{X}\mapsto\mathbb{R}_{>0},
  \forall g:\mathcal{Y}\mapsto\mathbb{R}_{>0},$ we have
  $\mathbb{E}f(X)g(Y)\geq ||f(X)||_{p_1}||g(Y)||_{p_2}.$
\end{itemize}
\vspace{0.1in} This prompts a natural extension to $k$-random
variables using the $k$-random variable H\"{o}lder inequalities.  The
most general H\"{o}lder and reverse H\"{o}lder inequalities for $k$
random variables are respectively given by:
\begin{align}
  \label{eq:general-Holder}
  \mathbb{E}\Pi_{i=1}^k W_i \leq \Pi_{i=1}^k ||W_i||_{p_i},& \qquad
  p_i> 1, \sum_{i=1}^k \frac{1}{p_i}= 1; \\
  \label{eq:general-reverse-Holder}
  \mathbb{E}\Pi_{i=1}^k W_i \geq \Pi_{i=1}^k ||W_i||_{p_i},& \qquad
  p_i< 1, p_i\neq 0,\mbox{exactly one $p_i>0$},\sum_{i=1}^k
  \frac{1}{p_i}= 1, W_i\geq 0.
\end{align}

\begin{proof}[Proof of H\"{o}lder and reverse H\"{o}lder inequalities]
  By the weighted arithmetic mean-geometric mean inequality, we have
  for any real numbers $y_1,y_2,\ldots, y_k\geq 0,$ and
  $p_1,p_2,\ldots, p_k>1$ satisfying $\sum_{i=1}^k \frac{1}{p_i} = 1,$
  \begin{align}
    \label{eq:weighted-AM-GM}
    \Pi_{i=1}^k y_i \leq \sum_{i=1}^k \frac{y_i^{p_i}}{p_i}.
  \end{align}
Setting $y_i = \frac{|W_i|}{||W_i||_{p_i}}$ and taking expectations
gives the H\"{o}lder inequality.

Now, if $0<p_1<1, p_2,p_3,\ldots, p_k<0$ satisfying $\sum_{i=1}^k
\frac{1}{p_i} = 1,$ we may set $q_1=\frac{1}{p_1},
q_i=\frac{-p_i}{p_1}, i=2,3,\ldots,k,$ so that $q_i>1$ and
$\sum_{i=1}^k \frac{1}{q_i} = 1.$ Using \eqref{eq:weighted-AM-GM} with
$q_i$'s, we get 
\begin{align}
  \label{eq:reformulated}
  \Pi_{i=1}^k y_i \leq p_1y_1^{\frac{1}{p_1}} + \sum_{i=2}^k
  \frac{p_1}{-p_i}y_i^{-\frac{p_i}{p_1}}.
\end{align}
For any $x_1,x_2,\ldots,x_k>0,$ choose $y_1 = \Pi_{i=1}^k
x_i^{p_1}$ and $y_i = x_i^{-p_1}$ for $i=2,3,\ldots,k,$ to get 
\begin{align}
  \label{eq:re-reformulated}
\sum_{i=1}^k \frac{x_i^{p_i}}{p_i} \leq \Pi_{i=1}^k x_i.
\end{align}
Setting $x_i = \frac{|W_i|}{||W_i||_{p_i}}$ and taking expectations
proves the reverse H\"{o}lder inequality for $W_i>0$ almost surely,
$i=1,2,\ldots, k.$ If $W_i\geq 0,$ we can set $W_i^\prime =
W_i+\epsilon$ and let $\epsilon\downarrow 0$ to complete the proof.
\end{proof}

\begin{remark}
  Both H\"{o}lder and reverse H\"{o}lder inequalities can also be
  proved by recursively invoking the inequalities for two
  variables. As a demonstration, fix any $0<p,q<1.$ For any
  non-negative real-valued $W_1,W_2,W_3,$
  \begin{align*}
    \mathbb{E}W_1W_2W_3 & \geq ||W_1W_2||_p |W_3|_{\frac{-p}{1-p}} \\
    & = \left(\mathbb{E}(W_1W_2)^p\right)^{\frac{1}{p}}
      |W_3|_{\frac{-p}{1-p}} \\
      & \geq \left(||W_1^p||_q||W_2^p||_{\frac{-q}{1-q}}\right)^{\frac{1}{p}}
      |W_3|_{\frac{-p}{1-p}} \\
      & = ||W_1||_{pq} ||W_2||_{\frac{-pq}{1-q}} ||W_3||_{\frac{-p}{1-p}}.
  \end{align*}

  It is easy to check that any reverse H\"{o}lder inequality may be
  obtained in this way by suitable choice of $p,q.$
\end{remark}

\begin{remark}
  The reverse H\"{o}lder inequality will also hold if some of the
  $p_i$ were equal to zero as long as the point $(p_1,p_2,\ldots,p_k)$
  is the limit of points satisfying $p_i\leq 1, p_i\neq
  0,\mbox{exactly one $p_i>0$},\sum_{i=1}^k \frac{1}{p_i}= 1.$ In
  particular, if we set for any integer $M>1,$
  $p_1^{(M)}=\frac{1}{Mk}, p_2^{(M)}=p_3^{(M)}=\ldots=p_k^{(M)} =
  -\frac{k-1}{Mk-1},$ then $(p_1^{(M)},p_2,^{(M)},\ldots,p_k^{(M)})$
  is a legitimate choice for the reverse H\"{o}lder's
  inequality. Taking the limit as $M\to\infty,$ we get the inequality
  $\mathbb{E}\Pi_{i=1}^k W_i \geq \Pi_{i=1}^k ||W_i||_{0},$ which is
  also valid for all random variables $W_i\geq 0$ and is a reverse
  H\"{o}lder's inequality.
\end{remark}

\begin{remark}
  The restriction in reverse H\"{o}lder inequality that exactly one
  $p_i>0$ is necessary. If no such $p_i$ exists, then the inequality
  is a consequence of $\mathbb{E}\Pi_{i=1}^k W_i \geq \Pi_{i=1}^k
  ||W_i||_{0}$ in the previous remark and the montonicity of norms. On
  the other hand, if more than one such $p_i$ exists, say $p_1,p_2>0,$
  then we can choose any mutually exclusive events $A,B$ such that
  $P(A\cap B) = 0, P(A)>0, P(B)>0.$ Set $W_1 = 1_A, W_2=1_B,
  W_3=W_4=\ldots, W_k=1.$ The reverse H\"{o}lder inequality, if true,
  would then yield $P(A\cap B)\geq
  P(A)^{\frac{1}{p_1}}P(B)^{\frac{1}{p_2}}$ which is false.
\end{remark}

Define $(p_1,p_2,\ldots,
p_k)\in\mathcal{H}(X_1;X_2;\ldots;X_k)$ if 
\begin{itemize}
\item $p_1,p_2, \ldots, p_k\geq 1,$ and $\forall
  f_i:\mathcal{X}_i\mapsto\mathbb{R}, i=1,2,\ldots, k$ we have
  $$\mathbb{E}\Pi_{i=1}^k f_i(X_i)\leq \Pi_{i=1}^k||f_i(X_i)||_{p_i};$$
\item $p_1,p_2, \ldots, p_k\leq 1,$ and $\forall
  f_i:\mathcal{X}_i\mapsto\mathbb{R}_{>0}, i=1,2,\ldots, k$ we have
  $$\mathbb{E}\Pi_{i=1}^k f_i(X_i)\geq
  \Pi_{i=1}^k||f_i(X_i)||_{p_i};$$
\end{itemize}

\begin{remark}
  The restriction to the orthant $p_1,p_2, \ldots, p_k\geq 1$ for the
  forward H\"{o}lder contraction is without loss of generality:
  Assuming $X_1$ is a non-constant random variable and $f_1$ is chosen
  so that $f_1(X_1)$ is non-constant and $f_2,f_3,\ldots, f_k$ are
  chosen to be constants, the inequality will hold only if $p_1\geq
  1.$ Likewise, the restriction to the orthant $p_1,p_2, \ldots,
  p_k\leq 1,$ for the reverse H\"{o}lder contraction is without loss
  of generality.
\end{remark}

It is easy to check that tensorization, data processing and
appropriate semi-continuity properties continue to hold for
$\mathcal{H}(X_1;X_2;\ldots;X_k)$ so we have the following
observation.

\vspace{0.1in}
\begin{adjustwidth}{0.6in}{0.6in}
\begin{observation}
  Non-interactive simulation of $(U_1,U_2,\ldots, U_k)\sim
  Q(u_1,u_2,\ldots, u_k)$ using $(X_1,X_2,\ldots, X_k)\sim
  P(x_1,x_2,\ldots, x_k)$ is possible only if, for all non-empty
  subsets $S_1,S_2,\ldots, S_m\subseteq \{1,2,\ldots, k\},$
  $\mathcal{H}(X_{S_1};X_{S_2};\ldots; X_{S_m})\subseteq
  \mathcal{H}(U_{S_1}; U_{S_2};\ldots; U_{S_m}).$
\end{observation}
\end{adjustwidth}
\vspace{0.1in}




Similarly, using maximal correlation, we can make the following observation:

\vspace{0.1in}
\begin{adjustwidth}{0.6in}{0.6in}
\begin{observation}
  Non-interactive simulation of $(U_1,U_2,\ldots, U_k)\sim
  Q(u_1,u_2,\ldots, u_k)$ using $(X_1,X_2,\ldots, X_k)\sim
  P(x_1,x_2,\ldots, x_k)$ is possible only if for all non-empty
  subsets $S_1,S_2\subseteq \{1,2,\ldots, k\},$ we have
  $\rho_m(X_{S_1};X_{S_2}) \geq \rho_m(U_{S_1};U_{S_2}).$
\end{observation}
\end{adjustwidth}
\vspace{0.1in}



\begin{example}
  We define the following distributions of \emph{DSBS triples} as
  shown in Fig.~\ref{fig:DSBStriple}. For chosen $0\leq
  \epsilon_X,\epsilon_Y, \epsilon_Z<\frac 12,$ we define $(X,Y,Z)\sim
  \DSBStriple(\epsilon_X,\epsilon_Y,\epsilon_Z)$ as the unique triple
  joint distribution satisfying $(Y,Z)\sim\DSBS(\epsilon_X),
  (X,Y)\sim\DSBS(\epsilon_Z), (X,Z)\sim\DSBS(\epsilon_Y)$ (note that
  there are two such distributions if
  $\epsilon_X=\epsilon_Y=\epsilon_Z=\frac{1}{2}).$ Such a distribution
  exists as long as the triangle inequalities
  $\epsilon_X+\epsilon_Y\geq \epsilon_Z, \epsilon_X+\epsilon_Z\geq
  \epsilon_Y, \epsilon_Z+\epsilon_Y\geq \epsilon_X$ are satisfied and
  the joint distribution of $(X,Y,Z)$ is given by:

  \begin{figure}[h]
    \begin{center}
      \includegraphics[width=3in, height=!]{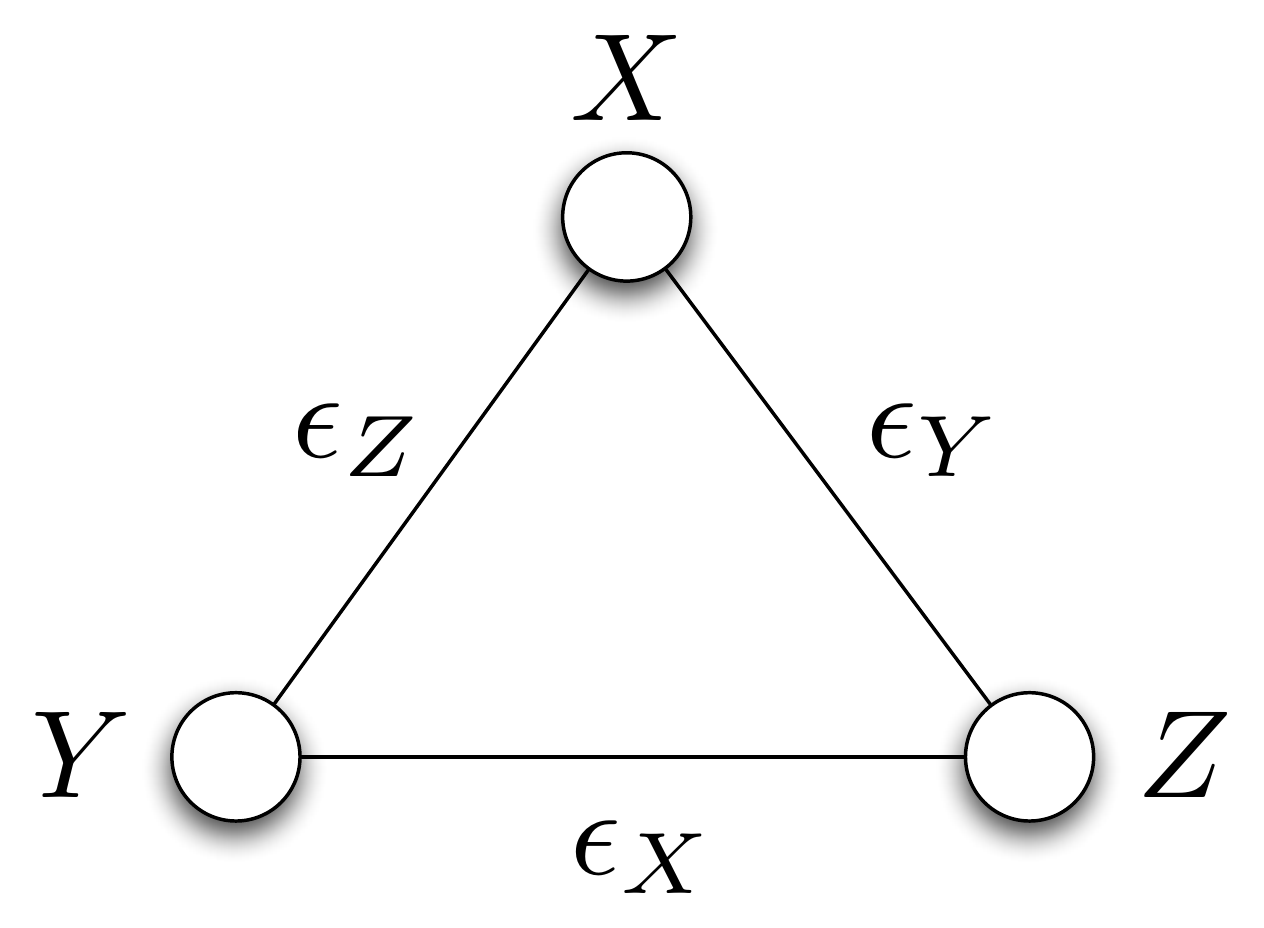}
      \caption{$(X,Y,Z)\sim \DSBStriple(\epsilon_X,\epsilon_Y,\epsilon_Z)$}
      \label{fig:DSBStriple}
    \end{center}
  \end{figure}

  \begin{align}
    P_{X,Y,Z}(0,0,0) = P_{X,Y,Z}(1,1,1) & = 
    \frac{2-\epsilon_X-\epsilon_Y-\epsilon_Z}{4} \\
    P_{X,Y,Z}(0,0,1) = P_{X,Y,Z}(1,1,0) & =
    \frac{\epsilon_X+\epsilon_Y-\epsilon_Z}{4} \\
    P_{X,Y,Z}(0,1,0) = P_{X,Y,Z}(1,0,1) & =
    \frac{\epsilon_X-\epsilon_Y+\epsilon_Z}{4} \\
    P_{X,Y,Z}(0,1,1) = P_{X,Y,Z}(1,0,0) & = 
    \frac{-\epsilon_X+\epsilon_Y+\epsilon_Z}{4}.
  \end{align}

  If either $A$ or $B$ is binary-valued, then one can simply write
  \cite{Witsenhausen75}
  \begin{align}
    \rho_m^2(A;B) = -1+\sum_{a,b} \frac{p_{A,B}(a,b)^2}{p_{A}(a)p_{B}(b)}.
  \end{align}
  Using this simple formula, we find that the various maximal
  correlation terms for $(X,Y,Z)\sim\DSBStriple(\epsilon_X,\epsilon_Y,
  \epsilon_Z)$ are given by:

  \begin{align}
    \rho_m(X;Y)  & = 1-2\epsilon_Z, \label{eq:start}\\
    \rho_m(X;Y,Z) & =
    \sqrt{\frac{(\epsilon_Y-\epsilon_Z)^2}{\epsilon_X}+\frac{(1-\epsilon_Y-\epsilon_Z)^2}{1-\epsilon_X}}.\label{eq:end}
  \end{align}

  Now, consider the following three-agent non-interactive simulation
  problem. Agents Alice, Bob, and Charlie observe $X^n, Y^n, Z^n$
  respectively and output (as a function of their observations and
  their private randomness) $\tilde{U}, \tilde{V}, \tilde{W}$
  respectively, which is required to be close in total variation to
  the target distribution $(U,V,W)$ as shown in Fig.~\ref{fig:three}.

  \begin{figure}[h]
    \begin{center}
      \includegraphics[width = 2in,
      height=!]{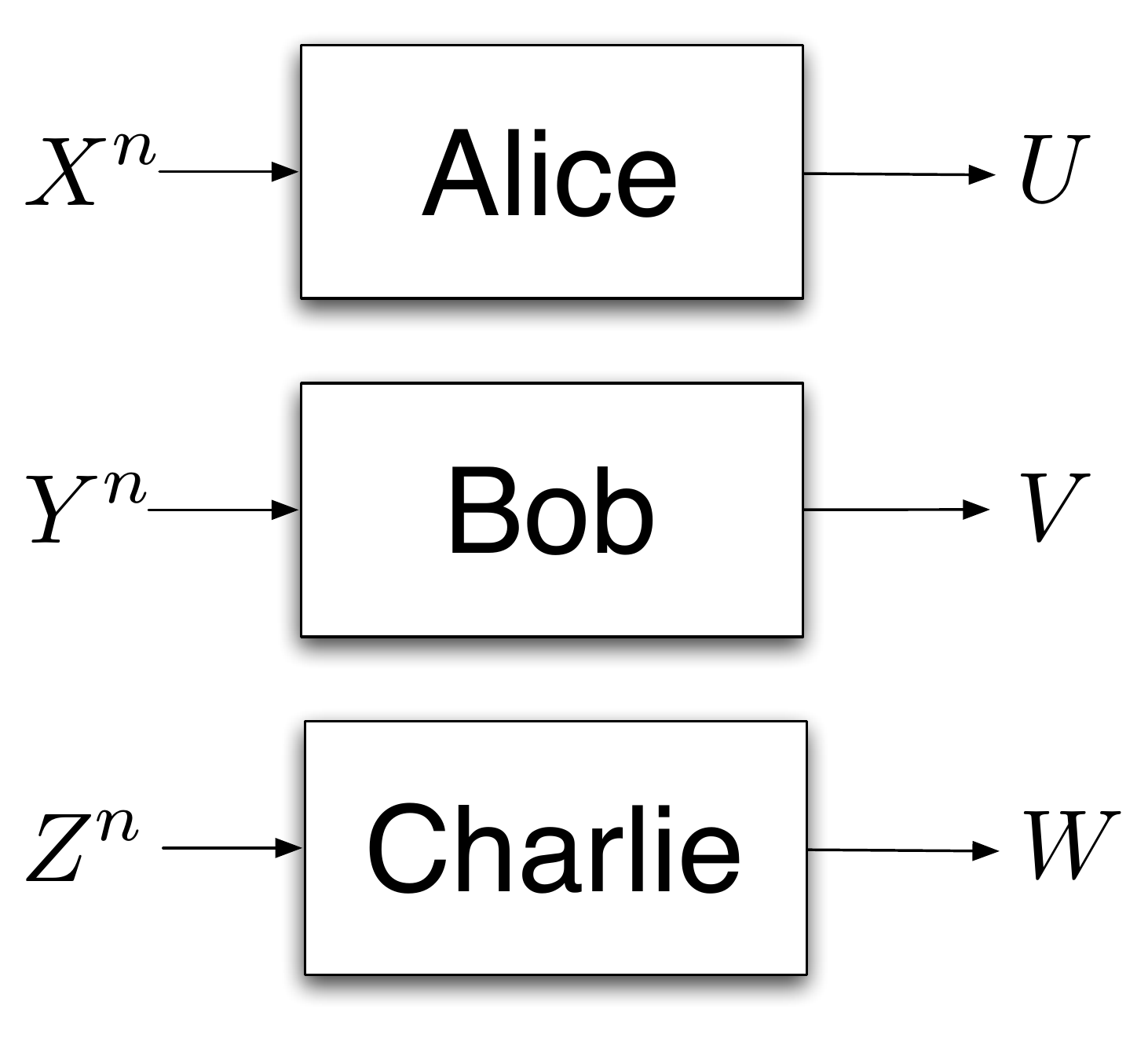}
      \caption{Three-user non-interactive simulation problem}
      \label{fig:three}
    \end{center}
  \end{figure}

  Suppose that for some $\epsilon<\frac{1}{2},$ the source and target
  distributions are specified by
  $(X,Y,Z)\sim\DSBStriple(\epsilon,\epsilon,\epsilon)$ and
  $(U,V,W)\sim\DSBStriple(\epsilon,2\epsilon(1-\epsilon),\epsilon)$ as
  shown in Fig.~\ref{fig:example}. In Section~\ref{subsec:example1},
  we pointed out that for a two-agent problem, non-interactive
  simulation of a DSBS target distribution with parameter
  $\beta<\frac{1}{2}$ using a DSBS source distribution with parameter
  $\alpha<\frac{1}{2}$ is possible if and only if the target
  distribution is more noisy, i.e. $\alpha\leq \beta.$ Thus, for this
  example, each pair of agents can perform the marginal pair
  simulation desired of them. However, the three agents cannot
  simulate the desired triple joint distribution.

\begin{figure}[h]
  \begin{center}
    \subfigure[Source distribution]{\includegraphics[viewport=10 8 354 250, width=1.6in, height=!]{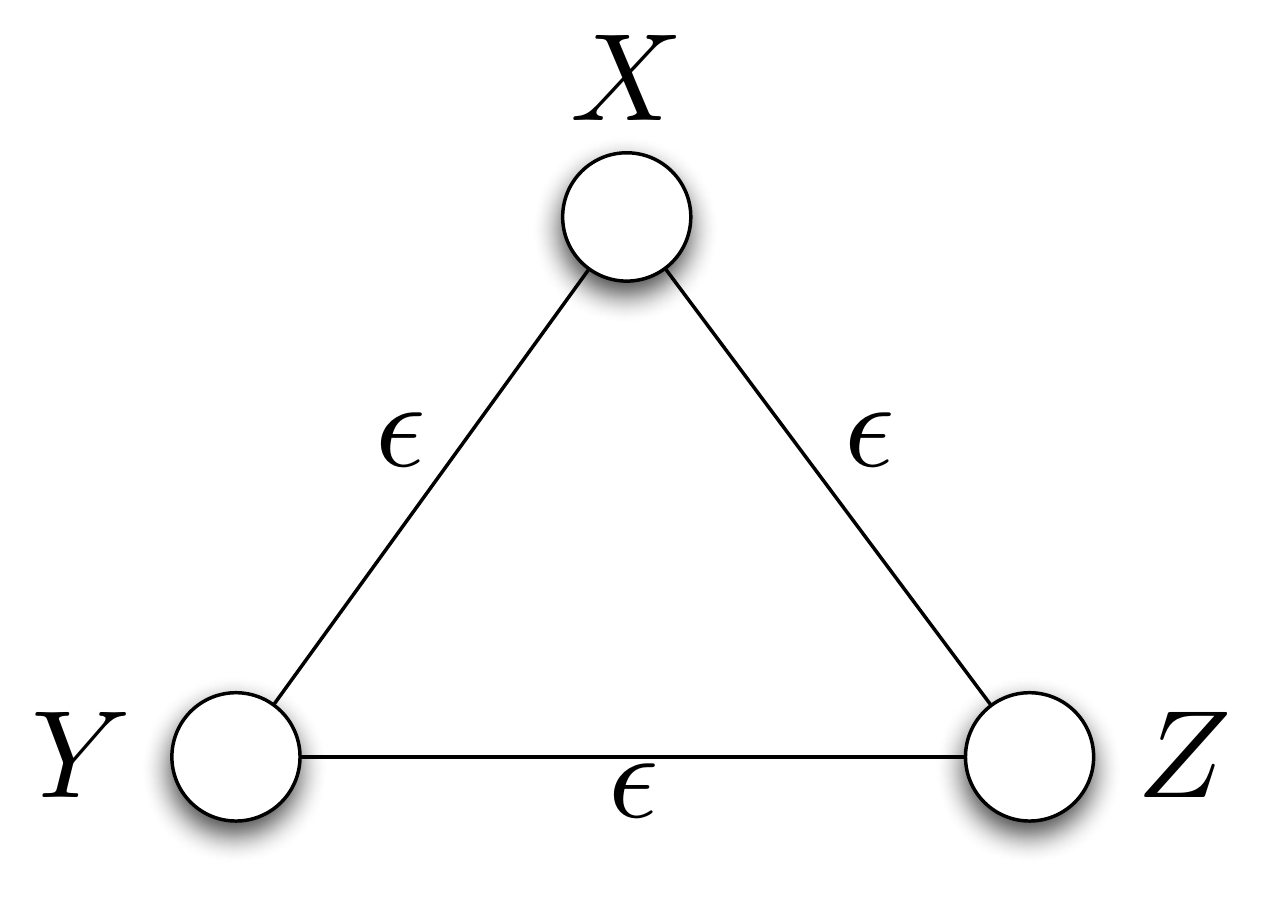}}\hspace{0.1in}
    \subfigure[Target distribution]{\includegraphics[viewport=10 10 360 255, width=1.6in, height=!]{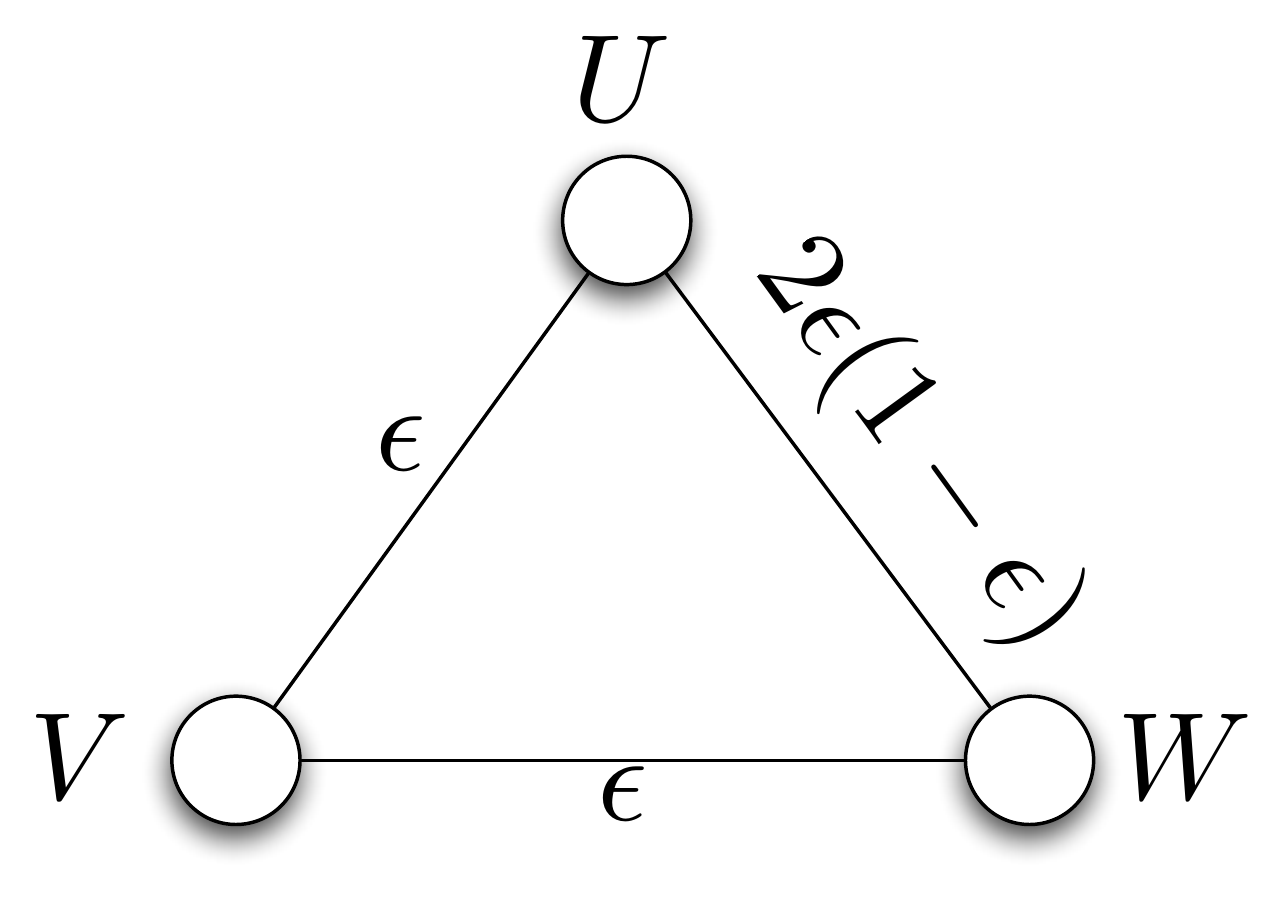}}
    \caption{Three random variable simulation example: Every pair of
      agents can achieve the desired simulation but the triple cannot.}
    \label{fig:example}
  \end{center}
\end{figure}


Using the formula \eqref{eq:end}, we get
\begin{align}
\rho_m(X,Z;Y) & = \frac{1-2\epsilon}{\sqrt{1-\epsilon}}\, , \\
\rho_m(U,W;V) & = \frac{1-2\epsilon}{\sqrt{1-2\epsilon+2\epsilon^2}}\,
.
\end{align}
For $0<\epsilon<\frac{1}{2},$ we have
$1-2\epsilon+2\epsilon^2<1-\epsilon,$ which gives
$\rho_m(X,Z;Y)<\rho_m(U,W;V).$ This shows that even if agents Alice
and Charlie were to combine their observations and their random
variable generation tasks to form one agent Alice-Charlie, then
Alice-Charlie and Bob cannot achieve the desired non-interactive
simulation.
\end{example}

\begin{example}
  Consider the following choices of source distribution $P(x,y,z)$ and
  target distribution $Q(u,v,w).$ 
  
  \begin{align*}
  P(x,y,z) = \begin{cases} 
    a_0 & \mbox{ if } (x,y,z) = (0,0,0), \\
    a_2 & \mbox{ if } (x,y,z) = (0,1,1), (1,0,1), (1,1,0),
\end{cases}
\end{align*}

where
\begin{align}
  \label{eq:0}
   a_0+3a_2=1,
\end{align}
 i.e. $(X,Y,Z)$ take values on the 4 sequences that
satisfy $X\oplus Y\oplus Z = 0$ (addition modulo 2).

\begin{align*}
  Q(u,v,w) = \begin{cases} 
    b_0 & \mbox{ if } (u,v,w) = (0,0,0),\\
    b_1 & \mbox{ if } (u,v,w) = (0,0,1), (0,1,0), (1,0,0),\\
    b_2 & \mbox{ if } (u,v,w) = (0,1,1), (1,0,1), (1,1,0),\\
    b_3 & \mbox{ if } (u,v,w) = (1,1,1).
  \end{cases}
\end{align*}

We will choose these parameters so that for some $0<\gamma<1,$ we have 

\begin{align}
  b_0 + b_1 & = a_0 + 2a_2\gamma + a_2\gamma^2,\label{eq:1} \\
  b_1 + b_2 & = a_2(1-\gamma^2),  \label{eq:2}\\
  b_2 + b_3 & = a_2(1-\gamma)^2~. \label{eq:3}
\end{align}

Consider the question of whether $(U, V, W)$ can be simulated from
$(X, Y, Z).$ For simulation of pair $(U,V)$ from $(X,Y),$ note that if
$A_1, A_2\sim\Ber(\gamma)$ i.i.d. and mutually independent of $(X,Y),$
then $$\left(X\oplus (A_1\cdot 1_{X=1}), Y\oplus (A_2\cdot 1_{Y=1})\right)\stackrel{d}{=} (U,V)$$ because of conditions
\eqref{eq:1}, \eqref{eq:2}, \eqref{eq:3}. By symmetry then, every pair
of agents can achieve the desired simulation.

Now, if we imagine two agents observe $(X,Y)$ and $Z$ respectively and
are required to simulate $(U,V)$ and $W$ respectively, then again this
is possible since $(X,Y)$ uniquely determines $Z,$ so the agents now
have access to shared randomness which can be used to generate any
required joint distribution.

However, consider the specific choice:
\begin{align*}
  a_0 = 0.825, \gamma = 0.2, b_0 = 0.8,
\end{align*}
so that the other parameters are fixed from \eqref{eq:0},
\eqref{eq:1}, \eqref{eq:2}, \eqref{eq:3} to be:
\begin{align*}
a_2=0.058333...,\qquad b_1 = 0.0506666...,\qquad b_2 = 0.005333...,\qquad
  b_3 = 0.032~.
\end{align*}

Here, we find computationally that
\begin{align}
  \kappa & :=\inf\{p\geq 1: (p,p,p)\in\mathcal{H}(X;Y;Z)\} = 1.93...;   \label{eq:kappa}\\
  \zeta & :=\inf\{p\geq 1: (p,p,p)\in\mathcal{H}(U;V;W)\} = 2.07....~.    \label{eq:zeta}
\end{align}

We present numerical evidence supporting the above
claims. Specifically, we will show that $1.85<\kappa<1.95$ and
$\zeta>2.05.$

Using H\"{o}lder's inequality, it is easy to verify that the following
two statements are equivalent:
\begin{align}
\mathbb{E}f(X)g(Y)h(Z) & \leq ||f(X)||_p||g(Y)||_p||h(Z)||_p,\ \forall
  f:\mathcal{X}\to\mathbb{R}, g:\mathcal{Y}\to\mathbb{R},
  h:\mathcal{Z}\to\mathbb{R},  \\
  ||\mathbb{E}[f(X)g(Y)|Z]||_{p^\prime} & \leq ||f(X)||_p ||g(Y)||_p,\
                                          \forall
                                          f:\mathcal{X}\to\mathbb{R},
                                          g:\mathcal{Y}\to\mathbb{R},
\end{align}
and furthermore, equivalently, all functions above may have co-domain
$\mathbb{R}_{\geq 0}.$ We choose $f(x) = (1+f)1_{x=1}+(1-f)1_{x=0}$
and $g(y) = (1+g)1_{y=1} + (1-g)1_{y=0}.$ It suffices to consider
functions of this form since the inequalities above are
homogeneous. Fig.~\ref{fig:proof_holder_contractivity} shows contour
plots of the ratio
$\frac{||\mathbb{E}[f(X)g(Y)|Z]||_{p^\prime}}{||f(X)||_p ||g(Y)||_p}$
where the X-axis represents the variable $f\in [-1,1]$ and the Y-axis
represents the variable $g\in [-1,1].$ For $p=1.95,$ the ratio is
upper-bounded by 1, whereas for $p=1.85,$ the ratio takes the value
$1.0088...$ at $f=g=-1.$ (Note that the color bar in
Fig.~\ref{fig:proof_holder_contractivity} has a maximum value of $1.0$
for $p=1.95$ and a maximum value of a little greater than $1.0$ for
$p=1.85.$) Thus, $(1.95,1.95,1.95)\in\mathcal{H}(X;Y;Z)$ but
$(1.85,1.85,1.85)\not\in\mathcal{H}(X;Y;Z)$ and so,
$1.85<\kappa<1.95.$

\begin{figure}[t]
  \begin{center}
    \subfigure[$p = 1.95$]{\includegraphics[viewport=10 8 354 250, width=2.8in, height=!]{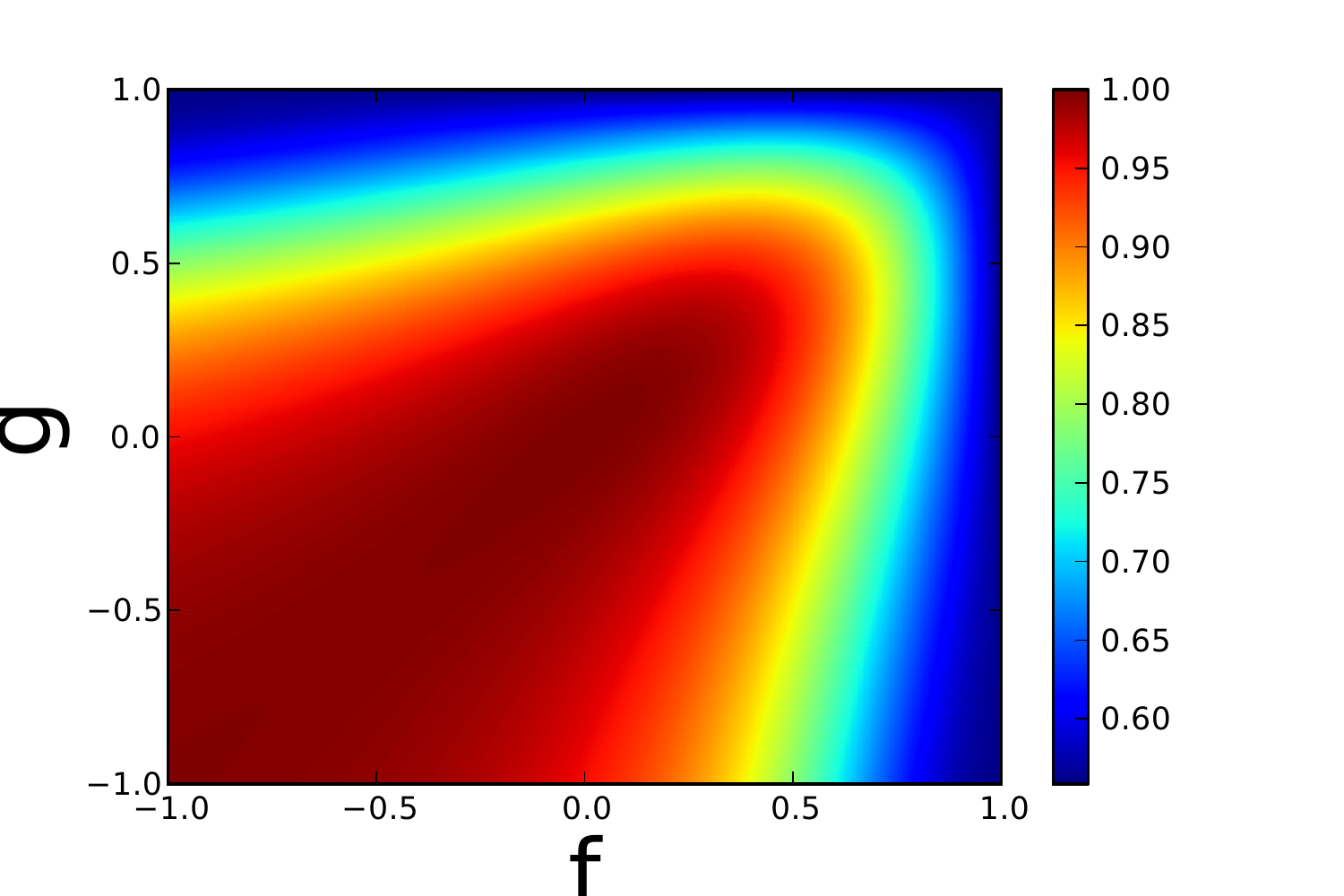}}\hspace{0.45in}
    \subfigure[$p = 1.85$]{\includegraphics[viewport=10 10 360 255, width=2.8in, height=!]{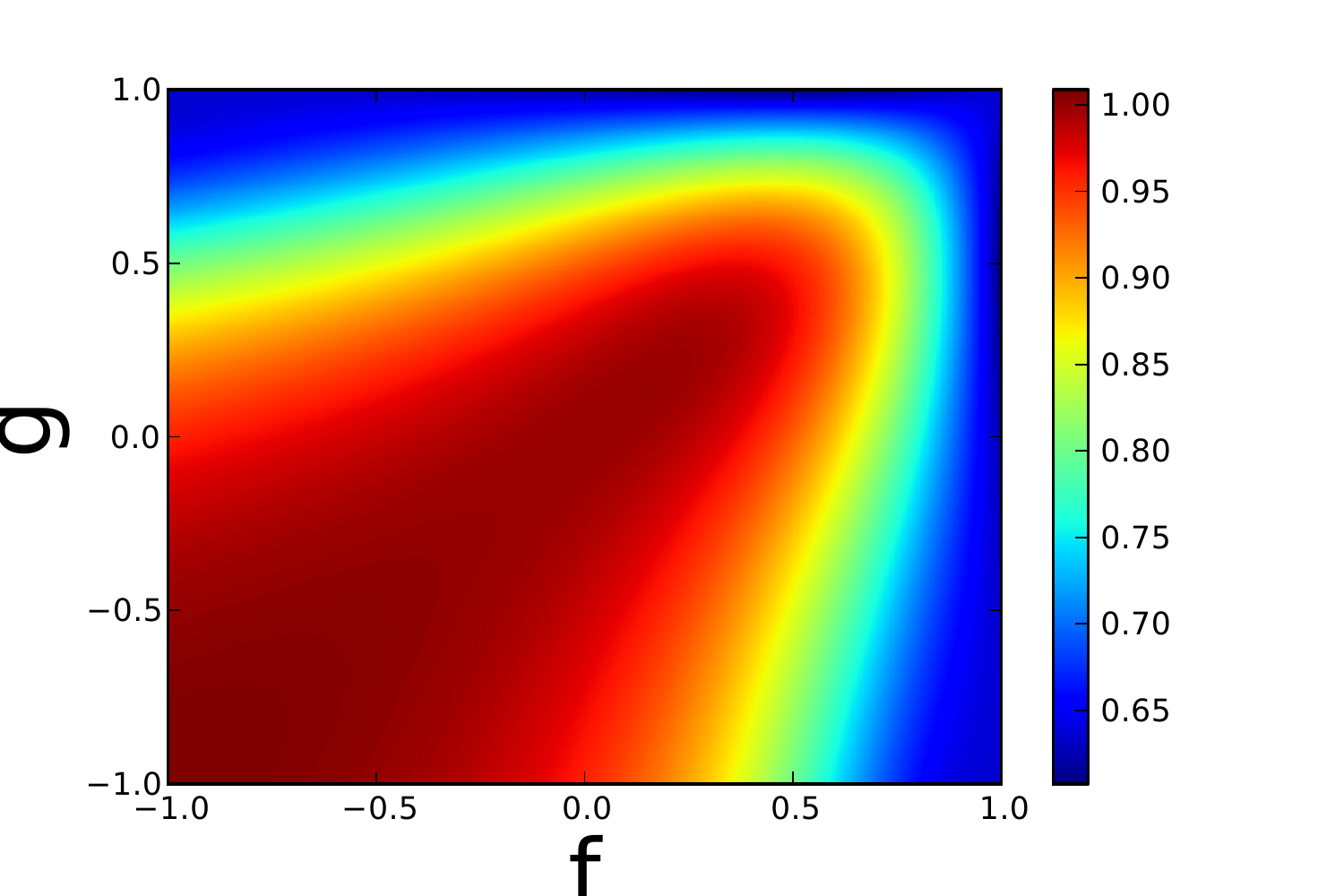}}
    \caption{Contour plots of the ratio
      $\frac{||\mathbb{E}[f(X)g(Y)|Z]||_{p^\prime}}{||f(X)||_p
        ||g(Y)||_p}$
      where $f(x) = (1+f)1_{x=1}+(1-f)1_{x=0}$ and
      $g(y) = (1+g)1_{y=1} + (1-g)1_{y=0}.$ The X-axis represents the
      variable $f\in [-1,1]$ and the Y-axis represents the variable
      $g\in [-1,1].$ We see numerically that for $p=1.95,$ the ratio
      is upper bounded by $1$ everywhere, but for $p=1.85,$ the ratio
      is maximized at $f=g=-1$ where it takes the value $1.0088...$
      This implies that $(1.95,1.95,1.95)\in\mathcal{H}(X;Y;Z)$ but
      $(1.85,1.85,1.85)\not\in\mathcal{H}(X;Y;Z).$}
    \label{fig:proof_holder_contractivity}
  \end{center}
\end{figure}

Now, consider the function
$\delta(\theta) = 9\cdot 1_{\theta=1} + 1_{\theta=0}.$ Then,

\begin{align}
  \mathbb{E}\delta(U) \delta (V) \delta (W) & = 26.792 \\
  ||\delta(U)||_{2.05}||\delta(V)||_{2.05}||\delta(W)||_{2.05} & = \left(||\delta(U)||_{2.05}\right)^3
                                                  = (2.9747...)^3 =
                                                  26.322... < 26.792.                                        
\end{align}
This proves that $(2.05,2.05,2.05)\not\in\mathcal{H}(U;V;W)$ and so,
$\zeta>2.05.$

Since $\kappa<1.95$ and $\zeta>2.05,$ the inclusion
$\mathcal{H}(X;Y;Z) \subseteq \mathcal{H}(U;V;W)$ is false and so, the
simulation of $(U,V,W)$ from $(X,Y,Z)$ is impossible.

\end{example}

\section{Acknowledgements}

We would like to thank Jingbo Liu, Elchanan Mossel and Vinod Prabhakaran for
useful discussions.

\vspace{0.1in} 

Research support from the ARO MURI grant W911NF-08-1-0233, ``Tools
for the Analysis and Design of Complex Multi-Scale Network'', from the
NSF grant CNS-0910702 and ECCS-1343398, from the NSF Science and Technology Center
grant CCF-0939370,``Science of Information'', from Marvell
Semiconductor Inc., and from the U.C. Discovery program is gratefully
acknowledged.

\bibliographystyle{IEEE} 
\bibliography{non_interactive_simulation_biblio}

\begin{thebibliography}{10}

\bibitem{AmazonPrimeAir}
{Amazon Prime Air},
\newblock \url{http://www.amazon.com/b?node=8037720011}.

\bibitem{RoboticsEngineering}
{Integration Innovation Inc.},
\newblock \url{http://www.i3-corps.com/environmental}.

\bibitem{GacsKorner72}
P.~G{\'a}cs and J.~K{\"o}rner,
\newblock ``Common information is far less than mutual information'',
\newblock {\em Problems of Control and Information Theory}, vol. 2, no. 2, pp.
  119--162, 1972.

\bibitem{Wyner75}
A.D. Wyner,
\newblock ``The common information of two dependent random variables'',
\newblock {\em IEEE Transactions on Information Theory}, vol. 21, no. 2, pp.
  163--179, March 1975.

\bibitem{HanVerdu93}
T.S. Han and S.~Verd{\'u},
\newblock ``{Approximation theory of output statistics}'',
\newblock {\em IEEE Transactions on Information Theory}, vol. 39, no. 3, pp.
  752--772, May 1993.

\bibitem{Cuff13}
P.~Cuff,
\newblock ``{Distributed channel synthesis}'',
\newblock {\em IEEE Transactions on Information Theory}, vol. 59, no. 11, pp.
  7071--7096, November 2013.

\bibitem{KumarLiElGamal14}
G.~Kumar, C.T. Li, and A.~El Gamal,
\newblock ``{Exact common information}'',
\newblock in {\em Proc. of IEEE ISIT}, Honolulu, Hawaii, July 2014.

\bibitem{Cuff08}
P.~Cuff,
\newblock ``Communication requirements for generating correlated random
  variables'',
\newblock in {\em Proc. of IEEE ISIT}, Toronto, Canada, July 2008.

\bibitem{Cuff10}
P.~Cuff, H.~Permuter, and T.~Cover,
\newblock ``Coordination capacity'',
\newblock {\em IEEE Transactions on Information Theory}, vol. 56, no. 9, pp.
  4181--4206, September 2010.

\bibitem{GohariAnantharam11}
A.A. Gohari and V.~Anantharam,
\newblock ``Generating dependent random variables over networks'',
\newblock in {\em Proceedings of the IEEE Information Theory Workshop}, Paraty,
  Brazil, October 2011, pp. 698--702.

\bibitem{Yassaee12}
M.H. Yassaee, A.A. Gohari, and M.R. Aref,
\newblock ``Channel simulation via interactive communications'',
\newblock in {\em Proc. of IEEE ISIT}, Cambridge, MA, July 2012.

\bibitem{AnantharamBorkar07}
V.~Anantharam and V.~Borkar,
\newblock ``Common randomness and distributed control: A counterexample'',
\newblock {\em Systems and Control Letters}, vol. 56, no. 7-8, pp. 568--572,
  July 2007.

\bibitem{Mossel04}
E.~Mossel, R.~O'Donnell, O.~Regev, J.~Steif, and B.~Sudakov,
\newblock ``{Non-interactive correlation distillation, inhomogeneous Markov
  chains, and the reverse Bonami-Beckner inequality}'',
\newblock {\em Israel Journal of Mathematics}, , no. 154, pp. 299--336, 2006.

\bibitem{Mossel10}
A.~Bogdanov and E.~Mossel,
\newblock ``On extracting common random bits from correlated sources'',
\newblock {\em IEEE Transactions on Information Theory}, vol. 57, no. 10, pp.
  6351 -- 6355, October 2011.

\bibitem{Witsenhausen75}
H.S. Witsenhausen,
\newblock ``On sequences of pairs of dependent random variables'',
\newblock {\em SIAM Journal on Applied Mathematics}, vol. 28, no. 1, pp.
  100--113, January 1975.

\bibitem{Bonami68}
Aline Bonami,
\newblock ``{Ensembles $\Lambda(p)$ dans le dual de $D^\infty$}'',
\newblock {\em Ann. Inst. Fourier}, vol. 18, no. 2, pp. 193--204, 1968.

\bibitem{Bonami70}
Aline Bonami,
\newblock ``{{\'E}tude des coefficients de {F}ourier des fonctions de
  $L^p(G)$}'',
\newblock {\em Ann. Inst. Fourier (Grenoble)}, vol. 20, no. 2, pp. 335--402,
  1970.

\bibitem{Nelson}
Edward Nelson,
\newblock ``{Construction of quantum fields from Markoff fields}'',
\newblock {\em J. Functional Analysis}, vol. 12, pp. 97--112, 1973.

\bibitem{Gross}
L.~Gross,
\newblock ``{Logarithmic Sobolev inequalities}'',
\newblock {\em Amer. J. Math.}, vol. 97, pp. 1061--1083, 1975.

\bibitem{Beckner}
William Beckner,
\newblock ``Inequalities in fourier analysis'',
\newblock {\em Ann. of Math.}, vol. 102, no. 1, pp. 159--182, 1975.

\bibitem{Borell}
Christer Borell,
\newblock ``Positivity improving operators and hypercontractivity'',
\newblock {\em Math. Z.}, vol. 180, no. 2, pp. 225--234, 1982.

\bibitem{Mossel11}
E.~Mossel, K.~Oleszkiewicz, and A.~Sen,
\newblock ``{On reverse hypercontractivity}'',
\newblock {\em Geometric and Functional Analysis}, vol. 23, no. 3, pp.
  1062--1097, 2011.

\bibitem{KKL88}
J.~Kahn, G.~Kalai, and N.~Linial,
\newblock ``{The influence of variables on Boolean functions}'',
\newblock in {\em Proc. of 29th Annual Symposium on Foundations of Computer
  Science}, 1988.

\bibitem{Friedgut98}
E.~Friedgut,
\newblock ``{Boolean functions with low average sensitivity}'',
\newblock {\em Combinatorica}, vol. 18, pp. 27--36, 1998.

\bibitem{MOO05}
E.~Mossel, R.~O'Donnell, and K.~Oleszkiewicz,
\newblock ``{Noise stability of functions with low influences{:} Invariance and
  Optimality}'',
\newblock in {\em Proceedings of the 46th Annual Symposium on Foundations of
  Computer Science}, 2005.

\bibitem{MosselRacz11}
E.~Mossel and M.~Racz,
\newblock ``{A quantitative Gibbard-Satterthwaite theorem without
  neutrality}'',
\newblock in {\em Proc. of the 44th Annual Symposium on Theory of Computing},
  2012.

\bibitem{AhlswedeGacs76}
R.~Ahlswede and P.~G{\'a}cs,
\newblock ``{Spreading of sets in product spaces and hypercontraction of the
  Markov operator}'',
\newblock {\em Annals of Probability}, vol. 4, pp. 925--939, 1976.

\bibitem{Raginsky13}
M.~Raginsky,
\newblock ``{Logarithmic Sobolev inequalities and strong data processing
  theorems for discrete channels}'',
\newblock in {\em Proc. of IEEE ISIT}, Istanbul, Turkey, 2013.

\bibitem{Polyanskiy13a}
Yury Polyanskiy,
\newblock ``{Hypothesis testing via a comparator and hypercontractivity}'',
\newblock {\em On webpage}, 2013,
  \url{http://people.lids.mit.edu/yp/homepage/data/htstruct_journal.pdf}.

\bibitem{Polyanskiy13b}
Y.~Polyanskiy,
\newblock ``{Hypercontractivity of spherical averages in Hamming space}'',
\newblock {\em preprint at arxiv: 1309.3014}, 2013.

\bibitem{AGKN_14}
V.~Anantharam, A.~Gohari, S.~Kamath, and C.~Nair,
\newblock ``On hypercontractivity and a data processing inequality'',
\newblock in {\em Proc. of IEEE ISIT}, Honolulu, Hawaii, July 2014.

\bibitem{AGKN_Allerton13}
V.~Anantharam, A.~Gohari, S.~Kamath, and C.~Nair,
\newblock ``{On hypercontractivity and the mutual information between Boolean
  functions}'',
\newblock in {\em Proc. of the 51st Annual Allerton Conference on
  Communications, Control and Computing}, Monticello, Illinois, October 2013.

\bibitem{Renyi59}
A.~R{\'e}nyi,
\newblock ``On measures of dependence'',
\newblock {\em Acta. Math. Acad. Sci. Hung.}, vol. 10, pp. 441--451, 1959.

\bibitem{HardyLittlewoodPolya}
G.H. Hardy, J.E. Littlewood, and G.~P{\'o}lya,
\newblock {\em Inequalities (2nd ed.)},
\newblock Cambridge University Press, Cambridge, 1952.

\bibitem{KumarCourtade13}
G.~Kumar and T.~Courtade,
\newblock ``{Which Boolean functions are most informative?}'',
\newblock in {\em Proc. of IEEE ISIT}, Istanbul, Turkey, 2013.

\bibitem{Kumar10}
Gowtham Kumar,
\newblock ``{On sequences of pairs of dependent random variables: A simpler
  proof of the main result using SVD}'',
\newblock {\em On webpage}, July 2010,
  \url{http://www.stanford.edu/~gowthamr/research/Witsenhausen_simpleproof.pdf}.

\bibitem{Billingsley}
P.~Billingsley,
\newblock {\em Probability and Measure (3rd ed.)},
\newblock Wiley, New York, 1995.

\end{thebibliography}

\appendix
\subsection{Proof of the claimed properties of $\rho_m$}
\label{subsec:rho_m_properties}

In this subsection, we prove the claimed properties of maximal correlation.

\begin{itemize}
\item (data processing inequality) For any functions $\phi, \psi,$
  $\rho_m(X;Y)\geq \rho_m(\phi(X),\psi(Y)).$

  \emph{Proof:} This is straightforward from the definition of $\rho_m.$

\item (tensorization) If $(X_1,Y_1)$ and $(X_2,Y_2)$ are independent, then
  $\rho_m(X_1,X_2;Y_1,Y_2)=\max\{\rho_m(X_1;Y_1),\rho_m(X_2;Y_2)\}.$

  \emph{Proof:} This property was shown by Witsenhausen
  \cite{Witsenhausen75}. The following exposition of Witsenhausen's
  proof is by Kumar \cite{Kumar10}. If we define
  $|\mathcal{X}|\times|\mathcal{Y}|$ matrices $P,Q$ by
  $P_{x,y} = P(x,y)$ and $Q_{x,y} = \frac{P(x,y)}{\sqrt{P(x)P(y)}},$
  then the top two singular values of $Q$ are $\sigma_1(Q)=1$ and
  $\sigma_2(Q) = \rho_m(X;Y)$ (for proof, see \cite{Kumar10}). The
  tensorization property then follows from the fact that the singular
  values of the tensor product of two matrices $A\otimes B$ are given
  by $\sigma_i(A)\sigma_j(B).$
  
  \item (Lower semi-continuity) If the space of probability
  distributions on $\mathcal{X}\times\mathcal{Y}$ is endowed with the
  total variation distance metric, then $\rho_m(X;Y)$ is a lower
  semi-continuous function of the joint distribution $P(x,y).$

  \emph{Proof:} Suppose $(X,Y), (X_1,Y_1), (X_2,Y_2),\ldots $ are
  random variable pairs taking values in the finite set
  $\mathcal{X}\times\mathcal{Y}$ satisfying
  $d_{\mathrm{TV}}((X_n,Y_n);(X,Y))\to 0$ as $n\to\infty.$ We will
  show that $\rho:=\liminf_{n\to\infty} \rho_m(X_n;Y_n) \geq
  \rho_m(X;Y).$ Let $\{j_n\}_{n=1}^\infty$ be a subsequence so that
  $\rho=\lim_{n\to\infty} \rho_m(X_{j_n};Y_{j_n}).$

  For any $\epsilon>0,$ there exists a $j(\epsilon)$ such that
  $\rho_m(X_{j_n};Y_{j_n})\leq \rho+\epsilon$ for all $j_n\geq
  j(\epsilon).$ Fix any functions $f:\mathcal{X}\mapsto\mathbb{R},
  g:\mathcal{Y}\mapsto\mathbb{R}$ such that
  $\mathbb{E}f(X)=\mathbb{E}g(Y)=0$ and $\mathbb{E}f(X)^2,
  \mathbb{E}g(Y)^2\leq 1.$ We will show $\mathbb{E}f(X)g(Y)\leq \rho$
  which will complete the proof.

  If $\mathbb{E}f(X)^2=0$ or $\mathbb{E}g(Y)^2=0,$ there is nothing to
  prove. So, suppose $\mathbb{E}f(X)^2,\mathbb{E}g(Y)^2>0.$ Since
  $\mathcal{X}\times\mathcal{Y}$ is a finite set,
  $d_{\mathrm{TV}}((X_{j_n},Y_{j_n});(X,Y))\to 0$ implies that
  $\Var(f(X_{j_n}))\to\Var(f(X))>0, \Var(g(Y_{j_n}))\to\Var(g(Y))>0.$
  There exists $j(f,g)$ such that $\Var(f(X_{j_n}))\geq
  \frac{\Var(f(X))}{2}, \Var(g(Y_{j_n})) \geq \frac{\Var(g(Y))}{2}$ for
  all $j\geq n(f,g).$

  Define for $j_n\geq \max\{j(\epsilon), j(f,g)\}$ the functions
  $f_{j_n}:\mathcal{X}\mapsto\mathbb{R}, g_{j_n}:\mathcal{Y}\mapsto\mathbb{R}$
  given by

  \begin{align}
    f_{j_n}(X)=\frac{f(X_{j_n})-\mathbb{E}f(X_{j_n})}{\sqrt{\Var(f(X_{j_n}))}}~, \\
    g_{j_n}(Y)=\frac{g(Y_{j_n})-\mathbb{E}g(Y_{j_n})}{\sqrt{\Var(g(Y_{j_n}))}}~,
  \end{align}
  which is possible since for such $j_n$ we have
  $\Var(f(X_{j_n})),\Var(g(Y_{j_n}))>0.$

  Again, we will have
  $\mathbb{E}f_{j_n}(X)g_{j_n}(Y)\to
  \frac{\mathbb{E}f(X)g(Y)}{\sqrt{\mathbb{E}f(X)^2\mathbb{E}g(Y)^2}}\geq
  \mathbb{E}f(X)g(Y).$
  But by definition, we have for $j_n\geq \max\{j(\epsilon), j(f,g)\}$
  that
  $\mathbb{E}f_{j_n}(X)g_{j_n}(Y)\leq \rho_m(X_{j_n};Y_{j_n})\leq
  \rho+\epsilon.$
  This gives $\mathbb{E}f(X)g(Y)\leq \rho+\epsilon.$ Since
  $\epsilon>0$ was arbitrary, we have $\mathbb{E}f(X)g(Y)\leq \rho.$

\end{itemize}

\subsection{Proof of the claimed properties of $s_p$}
\label{subsec:s_p_properties}

In this subsection, we prove the claimed properties of $s_p$ for
$p\neq 1.$

\begin{itemize}
\item (data processing inequality) For any functions $\phi, \psi,$
  $s_p(X;Y)\geq s_p(\phi(X);\psi(Y)).$

  \emph{Proof:} Let $W = \phi(X), Z = \psi(Y).$ Suppose for $1\leq
  q\leq p,$ we have $||\mathbb{E}[g(Y)|X]||_p\leq ||g(Y)||_q$ for all
  functions $g:\mathcal{X}\mapsto\mathbb{R}.$ For any function of $Z,$
  say $\theta(Z),$ we have

  \begin{align}
    ||\mathbb{E}[\theta(Z)|W]||_p & =
    ||\mathbb{E}[\theta(\psi(Y))|\phi(X)]||_p \\
    & \stackrel{(a)}{=} ||\mathbb{E}[\mathbb{E}[\theta(\psi(Y))|X]|\phi(X)]||_p \\
    & \stackrel{(b)}{\leq} ||\mathbb{E}[\theta(\psi(Y))|X]||_p \\
    & \leq ||\theta(\psi(Y))||_q \\
    & = ||\theta(Z)||_q,
  \end{align}
  where (a) follows from successive conditioning and (b) follows from
  Jensen's inequality: $||\mathbb{E}[A|\phi(X)]||_p\leq ||A||_p.$
  Similarly, we can deal with the case $1\geq q\geq p.$ This completes
  the proof.

\item (tensorization) If $(X_1,Y_1)$ and $(X_2,Y_2)$ are independent,
  then $s_p(X_1,X_2;Y_1,Y_2)=\max\{s_p(X_1;Y_1),s_p(X_2;Y_2)\}.$

  \emph{Proof:} Suppose $(X_1,Y_1)\sim P_1(x_1,y_1)$ and
  $(X_2,Y_2)\sim P_2(x_2,y_2)$ are both $(p,q)$-hypercontractive, with
  $p<1, p\neq 0.$ We remark that for the case of $p=0,$ we take limits
  in the standard way. Then,
  \begin{align}
    \mathbb{E}f(X_1)g(Y_1) & \geq ||f(X_1)||_{p^\prime}||g(Y_1)||_q \
    \forall\ f:\mathcal{X}_1\mapsto\mathbb{R}_{>0},\ \ \forall\ g:\mathcal{Y}_1\mapsto\mathbb{R}_{>0};\label{eq:X1Y1-hc}\\
    \mathbb{E}f(X_2)g(Y_2) & \geq ||f(X_2)||_{p^\prime}||g(Y_2)||_q \
    \forall\ f:\mathcal{X}_2\mapsto\mathbb{R}_{>0},\ \ \forall\ g:\mathcal{Y}_2\mapsto\mathbb{R}_{>0}.\label{eq:X2Y2-hc}    
  \end{align}
  
  Now, fix any positive-valued functions
  $f:\mathcal{X}_1\times\mathcal{X}_2\mapsto\mathbb{R}_{>0},
  g:\mathcal{Y}_1\times\mathcal{Y}_2\mapsto\mathbb{R}_{>0}.$

  \begin{align}
    \mathbb{E}f(X_1,X_2)g(Y_1,Y_2) & = \sum_{x_1,y_1}
    P_1(x_1,y_1) \sum_{x_2,y_2} P_2(x_2,y_2) f(x_1,x_2)g(y_1,y_2) \\
    & \stackrel{(a)}{\geq} \sum_{x_1,y_1} P_1(x_1,y_1)
    \left(\sum_{x_2}P_{X_2}(x_2)f(x_1,x_2)^{p^\prime}\right)^{\frac{1}{p^\prime}}\left(\sum_{y_2}P_{Y_2}(y_2)g(y_1,y_2)^{q}\right)^{\frac{1}{q}} \\
    & \stackrel{(b)}{\geq}
    \left(\sum_{x_1,x_2}P_{X_1}(x_1)P_{X_2}(x_2)f(x_1,x_2)^{p^\prime}\right)^{\frac{1}{p^\prime}}\left(\sum_{y_1,y_2}P_{Y_1}(y_1)P_{Y_2}(y_2)g(y_1,y_2)^{q}\right)^{\frac{1}{q}}\\
    & = ||f(X_1,X_2)||_{p^\prime}||g(Y_1,Y_2)||_q,
  \end{align}
  where (a) follows from \eqref{eq:X2Y2-hc} and (b) follows from
  \eqref{eq:X1Y1-hc}. This means $((X_1,X_2),(Y_1,Y_2))$ is
  $(p,q)$-hypercontractive. It is easy to see that if one of
  $(X_1,Y_1)$ or $(X_2,Y_2)$ is not $(p,q)$-hypercontractive, then
  $((X_1,X_2),(Y_1,Y_2))$ is not $(p,q)$-hypercontractive. Thus,
  $$q_p^*(X_1,X_2;Y_1,Y_2) = \min\{q_p^*(X_1;Y_1),q_p^*(X_2;Y_2)\},$$
  which gives $$s_p(X_1,X_2;Y_1,Y_2) =
  \max\{s_p(X_1;Y_1),s_p(X_2;Y_2)\}.$$

  For $p>1,$ the proof is similar; in this case, we
  find $$q_p^*(X_1,X_2;Y_1,Y_2) =
  \max\{q_p^*(X_1;Y_1),q_p^*(X_2;Y_2)\},$$ and $$s_p(X_1,X_2;Y_1,Y_2)
  = \max\{s_p(X_1;Y_1),s_p(X_2;Y_2)\}.$$
\item (lower semi-continuity) If the space of probability
  distributions on $\mathcal{X}\times\mathcal{Y}$ is endowed with the
  total variation distance metric, then $s_p(X;Y)$ is a lower
  semi-continuous function of the joint distribution $P(x,y).$

  \emph{Proof:} Let us fix $p<1.$ An identical proof holds for the
  case of $p>1.$ Suppose $(X,Y), (X_1,Y_1), (X_2,Y_2),\ldots $ are
  random variable pairs taking values in the finite set
  $\mathcal{X}\times\mathcal{Y}$ satisfying
  $d_{\mathrm{TV}}((X_n,Y_n);(X,Y))\to 0$ as $n\to\infty.$ Let $s :=
  \liminf_{n\to\infty} s_p(X_n;Y_n)\geq 0.$ We will show that $s\geq
  s_p(X;Y).$ Let $\{j_n\}_{n=1}^\infty$ be a subsequence so that
  $s=\lim_{n\to\infty} s_p(X_{j_n};Y_{j_n}).$

  We may assume without loss of generality that $s<1.$ For any
  $\epsilon>0,$ there exists a $j(\epsilon)$ such that
  $s_p(X_{j_n};Y_{j_n})\leq s+\epsilon$ for all $j_n\geq j(\epsilon).$
  We would like to show $s_p(X;Y)\leq s,$ i.e., that for any functions
  $f:\mathcal{X}\mapsto\mathbb{R}_{>0},
  g:\mathcal{Y}\mapsto\mathbb{R}_{>0},$ the following holds:
  \begin{align}\mathbb{E}f(X)g(Y) & \geq
    ||f(X)||_{p^\prime}||g(Y)||_{1+s(p-1)}. \end{align}
  For any given functions $f:\mathcal{X}\mapsto\mathbb{R}_{>0},
  g:\mathcal{Y}\mapsto\mathbb{R}_{>0},$ and any $j_n\geq j(\epsilon),$
  we have from $s_p(X_{j_n};Y_{j_n})\leq s+\epsilon$ that for $j_n\geq n(\epsilon),$
  \begin{align}\mathbb{E}f(X_{j_n})g(Y_{j_n}) & \geq
                                                ||f(X_{j_n})||_{p^\prime}||g(Y_{j_n})||_{1+(s+\epsilon)(p-1)}.
  \end{align} 
  
  From the portmanteau lemma \cite{Billingsley}, we get
  \begin{align}\mathbb{E}f(X)g(Y) & \geq
  ||f(X)||_{p^\prime}||g(Y)||_{1+(s+\epsilon)(p-1)}. \end{align}

Since this is true for each $\epsilon>0,$ we get from continuity of
$||.||_q$ in $q$ that
\begin{align}\mathbb{E}f(X)g(Y) & \geq ||f(X)||_{p^\prime}||g(Y)||_{1+s(p-1)}. \end{align}

Since this is true for any functions
$f:\mathcal{X}\mapsto\mathbb{R}_{>
  0},g:\mathcal{Y}\mapsto\mathbb{R}_{>0},$ we have $s_p(X;Y)\leq s.$

\begin{remark} Note that this implies that $q_p(X;Y)=1+s_p(X;Y)(p-1)$ is lower
  semi-continuous in the joint distribution for fixed $p>1$ and upper
  semi-continuous in the joint distribution for fixed $p<1.$
\end{remark}

\begin{remark}
  Lower semi-continuity of $\rho_m$ and $s_p$ was enough for our
  purposes. Indeed, $\rho_m$ and $s_p$ are not continuous in the
  underlying joint distribution. As an example, let $(X_n,Y_n)$ be
  binary-valued
  and have a joint probability distribution given by $\begin{bmatrix}\frac{1}{n} & 0 \\
    0 & 1-\frac{1}{n}\end{bmatrix}.$ Then, $(X_n,Y_n)\stackrel{d}{\to}
  (X,Y)$ where $(X,Y)$ has a joint probability distribution given by $\begin{bmatrix}0 & 0 \\
    0 & 1\end{bmatrix}.$ But $\rho_m(X_n;Y_n) = s_p(X_n;Y_n)= 1$ for
  each $n$ and each $p\neq 1,$ while $\rho_m(X;Y)=s_p(X;Y)=0.$

However, it may be shown that if $(X,Y)\sim P(x,y)$ satisfies the
assumption $P(x)>0\ \ \forall x\in\mathcal{X}, P(y)>0\ \ \ \forall
y\in\mathcal{Y},$ then $(X_n,Y_n)\stackrel{d}{\to} (X,Y)$ implies
$\lim_{n\to\infty} \rho_m(X_n;Y_n) = \rho_m(X;Y).$ To see this, use
the characterization $\rho_m(X;Y) = \sigma_2(A_{X;Y}),$ where the
matrix $A_{X;Y}$ is specified by $[A_{X;Y}]_{x,y} =
\frac{P(x,y)}{\sqrt{P(x)P(y)}}$ and $\sigma_2(\cdot)$ is the second
largest singular value \cite{Witsenhausen75, Kumar10}. Under the
assumption, $A_{X_n;Y_n}\to A_{X;Y}$ and the second largest singular
value is a continuous matrix functional.
\end{remark}

\end{itemize}

\subsection{Limiting properties of $s_p$: Proofs of
  Thm.~\ref{thm:limit_p=1} and Corollary~\ref{cor:limit_p=infty}}
\label{subsec:s*_properties}

  As in \cite{Mossel11}, we define for any non-negative random
  variable $X,$ the function
  $\Ent(X):=\mathbb{E}[X\log X] - \mathbb{E}[X]\cdot\log
  \mathbb{E}[X],$ where by convention $0\log 0:=0.$
  By strict convexity of the function $x\mapsto x\log x$ and Jensen's
  inequality, we get that $\Ent(X)\geq 0$ and equality holds if and only
  if $X$ is a constant almost surely. Also, we note that $\Ent(\cdot)$
  is homogenous, that is, $\Ent(aX) = a\Ent(X)$ for any $a\geq 0.$

  We begin by presenting first a simple lemma.

\begin{lemma}
  \label{lem:simplest}
For any random variable $Z$ satisfying $0\leq Z\leq K$ for some
constant $K>0$ and $\mathbb{E}Z=1$ and $0\leq u\leq 1,$ we have
\begin{align}
  1+ u\Ent(Z) - u^2L_1(K)
  \leq \|Z\|_{1+u} \leq 1 +
  u\Ent(Z)+ u^2 L_0(K),
\end{align}
where $L_0(K) = \frac 12\max\{K^u,1\}\max_{0\leq z\leq K} z(\log z)^2$
and $L_1(K) = (\max_{0\leq z\leq K} |z\log z|) + \frac12 (\max_{0\leq z\leq K} |z\log z|) ^2.$
\end{lemma}

\begin{proof}[Proof of Lemma~\ref{lem:simplest}]
  For any constant $0\leq u\leq 1$ and any $\theta\in\mathbb{R},$ a
  Taylor's series expansion yields
  
  $$1+u\theta \leq e^{u\theta} \leq 1+u\theta + \frac{u^2}{2} \theta^2\max\{e^{u\theta},1\}~.$$

  Thus, for any $0\leq z\leq K$ for some constant $K>0,$ and $0\leq
  u\leq 1,$ we have using $z^{1+u} = ze^{u\log z},$ 

  $$z+uz\log z \leq z^{1+u} \leq z+uz\log z + \frac{u^2}{2}z(\log
    z)^2\max\{z^u,1\}~.$$

    For any random variable $Z$ satisfying $0\leq Z\leq K$ almost
    surely and any $0\leq u\leq 1,$

    \begin{align}
      \mathbb{E}Z+u\mathbb{E}[Z\log Z] \leq \mathbb{E}[Z^{1+u}] & \leq \mathbb{E}Z+u\mathbb{E}[Z\log Z] + \frac{u^2}{2}\max\{K^u,1\}\mathbb{E}[Z(\log
                                                                  Z)^2]
                                                                  \nonumber
      \\
                                                                & \leq \mathbb{E}Z+u\mathbb{E}[Z\log Z] + u^2L_0(K)~.\label{eq:one_calc}
    \end{align}
    
    Now, again a Taylor's expansion yields that for $0\leq r\leq 1$
    and any $x\geq 0,$ we have

    \begin{align}
      1+rx -
      \frac{x^2}{2}r(1-r)
      \leq (1+x)^r \leq 1+rx~. \label{eq:two_calc}
    \end{align}

    Suppose $Z$ is any random variable that satisfies $0\leq Z\leq K$
    and $\mathbb{E}Z=1.$ Then $\mathbb{E}[Z\log Z] = \Ent(Z)\geq 0.$
    For any $0\leq u\leq 1,$ we get using the lower bounds in both
    \eqref{eq:one_calc} and \eqref{eq:two_calc} with the choice
    $r=\frac{1}{1+u}$ and $x = u\Ent(Z),$ 

    \begin{align*}
      1+ \frac{1}{1+u} u\Ent(Z) - \frac{u^2\Ent(Z)^2}{2}
      \frac{1}{1+u}\frac{u}{1+u} \leq \left(\mathbb{E}[Z^{1+u}]\right)^{\frac{1}{1+u}}.
    \end{align*}

    Similarly, using the upper bounds in both
    \eqref{eq:one_calc} and \eqref{eq:two_calc} with the choice
    $r=\frac{1}{1+u}$ and $x = u\Ent(Z)+u^2L_0(K),$ we get 
    \begin{align*}
    \left(\mathbb{E}[Z^{1+u}]\right)^{\frac{1}{1+u}} \leq 1 +
      \frac{1}{1+u} u\Ent(Z) +\frac{1}{1+u}u^2 L_0(K)~.
    \end{align*}

Putting the above two inequalities together,
\begin{align*}
      1+ \frac{1}{1+u} u\Ent(Z) - \frac{u^2\Ent(Z)^2}{2}
      \frac{1}{1+u}\frac{u}{1+u} \leq \|Z\|_{1+u} \leq 1 +
      \frac{1}{1+u} u\Ent(Z) +\frac{1}{1+u}u^2 L_0(K).
\end{align*}

Define $L_2(K) = \max_{0\leq z\leq K} |z\log z|$ and
observing that for $0\leq u\leq 1,$ we have $\frac 12 \leq
\frac{1}{1+u} \leq 1,$ we obtain

\begin{align*}
  1+ \frac{u\Ent(Z)}{1+u} - \frac{u^3}{2}L_2(K)^2
  \leq \|Z\|_{1+u} \leq 1 +
  \frac{u\Ent(Z)}{1+u} +u^2L_0(K).
\end{align*}

Further using the fact that for $0\leq u\leq 1,$ we have $1-u\leq
\frac{1}{1+u} \leq 1,$ we get

 \begin{align*}
  1+ u\Ent(Z) - u^2L_2(K) - \frac{u^3}{2}L_2(K)^2
  \leq \|Z\|_{1+u} \leq 1 +
  u\Ent(Z)+u^2L_0(K).
\end{align*}

Finally, since $L_1(K) = L_2(K) + \frac 12
L_2(K)^2$ and $u\leq 1,$ we have 
\begin{align}
  1+ u\Ent(Z) - u^2L_1(K) 
  \leq \|Z\|_{1+u} \leq 1 +
  u\Ent(Z)+u^2L_0(K). \label{eq:main-norm-bounds}
\end{align}

\end{proof}

Next, we present the proof of Thm.~\ref{thm:limit_p=1}.

\begin{proof}[Proof of Theorem~\ref{thm:limit_p=1} ]

  The theorem is easily seen to be true when $Y$ is a constant almost
  surely. We assume then that this is not the case and that $P_Y(y)>0$
  for all $y\in\mathcal{Y}$ and $P_X(x)>0$ for all $x\in\mathcal{X}$
  without loss of generality. Define
  $s := \sup \frac{\Ent(\mathbb{E}[g(Y)|X])}{\Ent(g(Y))},$ where the
  supremum is taken over functions
  $g:\mathcal{Y}\mapsto\mathbb{R}_{\geq 0}$ such that $g(Y)$ is not a
  constant almost surely.

For any distribution $R_Y(y)\not\equiv P_Y(y)$ consider the
non-constant non-negative valued function $g$ given by
$g(y) := \frac{R_Y(y)}{P_Y(y)}.$ This choice yields
$\Ent(g(Y))=D(R_Y(y)||P_Y(y))$ and
$\Ent(\mathbb{E}[g(Y)|X])=D(R_X(x)||P_X(x))),$ where
$R_X(x) = \sum_y \frac{P_{X,Y}(x,y)}{P_Y(y)}R_Y(y).$ Along with homogeneity
of $\Ent(\cdot),$ this means that $s= s^*(Y;X)$ and thus, from the
data processing inequality $0\leq s\leq 1.$

For non-negative $g,$ we always have
\begin{equation}\label{eqn:1-norm-equality}
  ||\mathbb{E}[g(Y)|X]||_1= ||g(Y)||_1\ \ \forall
  g:\mathcal{Y}\mapsto\mathbb{R}_{\geq 0}.
\end{equation}

Let $\mathcal{G}$ be the set of all non-negative functions
$g:\mathcal{Y}\mapsto\mathbb{R}_{\geq 0}$ that satisfy
$||g(Y)||_1=1.$ Note that for any $g\in\mathcal{G},$ both $g(Y)$ and
$\mathbb{E}[g(Y)|X]$ are bounded between $0$ and
$K:=\frac{1}{\min_{y} P_Y(y)}$ almost surely.

If $0\leq m\leq 1$ is any parameter satisfying $m<s,$ then
$(1+ \tau, 1+m\tau)\not\in\mathcal{R}(X;Y)$ for all sufficiently
small $\tau>0.$ To see this, fix $g_0$ to be any function in
$\mathcal{G}$ that satisfies
\begin{equation}\label{eqn:g_0}
  \frac{\Ent(\mathbb{E}[g_0(Y)|X])}{\Ent(g_0(Y))}\geq m+\frac{\delta}{2},
\end{equation}
where $\delta:=s-m.$ From Lemma~\ref{lem:simplest}, we have that for
any $g\in\mathcal{G},$
\begin{align}
  &1 +
    m\tau\Ent(g(Y)) - m^2\tau^2L_1(K) \leq ||g(Y)||_{1+m\tau} \leq 1+ m\tau\Ent(g(Y)) + m^2\tau^2 L_0(K),\label{eqn:m-tau}\\
  & 1 +
    \tau\Ent(\mathbb{E}[g(Y)|X]) - \tau^2L_1(K) \leq ||\mathbb{E}[g(Y)|X]||_{1+\tau} \leq 1+ \tau\Ent(\mathbb{E}[g(Y)|X]) + \tau^2 L_0(K).\label{eqn:tau}
\end{align}

Putting together \eqref{eqn:g_0}, \eqref{eqn:m-tau}, \eqref{eqn:tau},
we get the existence of $\tau_0>0$ such that
\begin{equation}
||\mathbb{E}[g_0(Y)|X]||_{1+\tau} > ||g_0(Y)||_{1+m\tau}\ \ \forall \tau: 
0<\tau\leq \tau_0.
\end{equation}

Thus, $s = s^*(Y;X) \geq \limsup_{p\to 1^+} s_p(X;Y) = \limsup_{p\to 1^+} \frac{q_p^*(X;Y)-1}{p-1}.$

If for some $0\leq m\leq 1$ we have $m>s,$ then define for any $g\in\mathcal{G},$
$$\tau(g) :=\max\{\zeta: 0\leq \zeta\leq 1,
||\mathbb{E}[g(Y)|X]||_{1+\eta} \leq ||g(Y)||_{1+m\eta} \mbox{ for all
} 0\leq \eta\leq \zeta\}.$$

From \eqref{eqn:1-norm-equality}, we have $\tau(g)\geq 0$ for all
$g\in\mathcal{G}.$

Let $g_1\in \mathcal{G}$ denote the constant function 1. Then,
$\tau(g_1) = 1.$ Lemma~\ref{lem:difficult} below shows that there is
an open neighborhood $U$ of $g_1$ in $\mathcal{G}$ and a constant
$\tau_0>0$ such that
$\tau(g) \geq \tau_0\ \forall g\in U.$

Over the compact set $\mathcal{G}\setminus U,$ we define 
$$\tau^\prime(g) :=\max\{\zeta: 0\leq \zeta\leq 1,
1 + \eta\Ent(\mathbb{E}[g(Y)|X]) +\eta^2L_0(K) \leq 1 + m\eta\Ent(g(Y)) -
m^2\eta^2L_1(K) \mbox{ for all } 0\leq \eta\leq \zeta\}.$$

Then, $\tau^\prime(g)\leq \tau(g)$ from Lemma~\ref{lem:simplest}. And
indeed,
$$\tau^\prime(g) = \min\left\{\frac{m\Ent(g(Y)) -
  \Ent(\mathbb{E}[g(Y)|X])}{L_0(K) + m^2L_1(K)},1\right\}.$$
Since $\tau^\prime(g)$ is continuous in $g$ over
$\mathcal{G}\setminus U,$ and furthermore strictly positive over that
set (since $m>s$ and because $\Ent(g(Y))>0$ for $g$ non-constant), we
have that $\tau^\prime$ attains its infimum over the compact set
$\mathcal{G}\setminus U.$ Since $\tau^\prime(g)\leq \tau(g),$ we also
have that $\inf_{g\in\mathcal{G}\setminus U} \tau(g)>0.$

Then,
$\inf_{g\in \mathcal{G}} \tau(g) =
\min\left\{\tau_0,\inf_{g\in\mathcal{G}\setminus U} \tau(g)\right\} >0.$
Using homogeneity of the norm, this establishes that
$(1+\tau,1+m\tau)\in\mathcal{R}(X;Y)$ for all $0\leq \tau\leq \tau_0$
for some $\tau_0>0$ and thus, that
$s = s^*(Y;X) \leq \liminf_{p\to 1^+} s_p(X;Y) = \liminf_{p\to 1^+}
\frac{q_p^*(X;Y)-1}{p-1}.$

Therefore,
$s = s^*(Y;X) = \lim_{p\to 1^+} s_p(X;Y) = \lim_{p\to 1^+}
\frac{q_p^*(X;Y)-1}{p-1}.$

Similarly, we can show the reverse hypercontractive case namely, that
$s = s^*(Y;X) = \lim_{p\to 1^-} s_p(X;Y) = \lim_{p\to 1^-}
\frac{q_p^*(X;Y)-1}{p-1}.$ This completes the proof of the theorem.

\begin{lemma}
  \label{lem:difficult}
  When $1\geq m>s,$ there exists an open neighborhood $U$ of the constant
  function $g_1$ in $\mathcal{G}$ and a constant $\tau_0>0$ such that
  $\tau(g) \geq \tau_0$ for all
  $g\in U.$
\end{lemma}

\begin{proof}[Proof of Lemma~\ref{lem:difficult}]
  Let $\mathcal{F}$ denote the set of all functions
  $f:\mathcal{Y}\mapsto\mathbb{R}$ such that $\mathbb{E}[f(Y)] = 0$
  and $\mathbb{E}[f(Y)^2] = 1.$ For any $f\in\mathcal{F},$ and any
  $y\in\mathcal{Y},$ we have $|f(y)|\leq \frac{1}{\min_{y}\sqrt{P_Y(y)}}.$

  For $0<\epsilon_0 < \frac 12\min_{y}\sqrt{P_Y(y)},$ the set
  $U(\epsilon_0):= \{g_1+\epsilon f: f\in\mathcal{F}, 0\leq
  \epsilon<\epsilon_0\}$
  is an open neighborhood of the constant function $g_1$ in
  $\mathcal{G}.$ Furthermore, $\frac 12\leq g(y)\leq \frac 32$ for all
  $y\in\mathcal{Y}$ and all $g\in U(\epsilon_0).$

  Let $m = (1+\delta)s$ where $s<1$ and $m\leq 1$ and where $\delta>0.$

  For $g\in\mathcal{G},$ denote $\chi_g(x) = \mathbb{E}[g(Y)|X=x]$ and note that
  $\frac 12\leq \chi_g(x)\leq \frac 32$ for all $x\in\mathcal{X}.$

  Now, for $0\leq \eta \leq 1,$
  \begin{align}
    \|g(Y)\|_{1+m\eta} & = \left(\sum_y P_Y(y)g(y)e^{m\eta \log
                         g(y)}\right)^{\frac{1}{1+m\eta}} \\
    & \geq e^{\frac{m\eta}{1+m\eta}\Ent(g(Y))} \label{eq:convexity}\\
    & \geq e^{\frac{m\eta}{s(1+m\eta)}\Ent(\mathbb{E}[g(Y)|X])} \label{eq:definition_of_s}\\
    & \geq e^{(1+\delta)\frac{\eta}{(1+\eta)}\Ent(\chi_g(X))} \\
    & \geq \left(1 + \eta(1+\delta)\Ent(\chi_g(X)) +
      \frac{\eta^2}{2}(1+\delta)^2\Ent(\chi_g(X))^2\right)^{\frac{1}{1+\eta}} \label{eq:exp_lower_bound},
  \end{align}

where \eqref{eq:convexity} follows from convexity of the exponential
function, \eqref{eq:definition_of_s} follows from the definition of
$s$ and \eqref{eq:exp_lower_bound} follows from $e^u\geq
1+u+\frac{u^2}{2}$ for $u\geq 0.$

Likewise, we have 
\begin{align}
  \|\mathbb{E}[g(Y)|X]\|_{1+\eta} & = \left(\sum_{x} P_X(x)
                                    \chi_g(x)e^{\eta\log
                                    \chi_g(x)}\right)^{\frac{1}{1+\eta}} \\
                                  & \leq \left(\sum_{x} P_X(x)
                                    \chi_g(x)\left(1+\eta\log \chi_g(x) +
                                    a\frac{\eta^2}{2}(\log \chi_g(x))^2\right)\right)^{\frac{1}{1+\eta}} \\
                                  & \leq \left(1+\eta\Ent(\chi_g(X)) +
                                    a\frac{\eta^2}{2}\sum_xP_X(x)\chi_g(x)(\log \chi_g(x))^2\right)^{\frac{1}{1+\eta}},
\end{align}
where $a>1$ is a constant such that $e^u \leq 1 + u + a\frac{u^2}{2}$
for $|u|\leq \log 2.$

Note that $\Ent(\chi_g(X)) = D(Q_X||P_X)$ where $Q_X(x) = P_X(x)\chi_g(x)$ for
all $x\in\mathcal{X}.$ By Pinsker's inequality, 

$$\Ent(\chi_g(X)) \geq \frac
12 \left(\sum_{x} |P_X(x)\chi_g(x) - P_X(x)|\right)^2.$$

Thus, for all $x\in\mathcal{X},$ we have

$$|\chi_g(x) - 1|\leq \frac{1}{\min_x P_X(x)}\sqrt{2\Ent(\chi_g(X))}.$$

If we define for $0\leq \alpha\leq 1,$ the function $\kappa(\alpha) := \max_{1-\alpha\leq v\leq 1+\alpha}
v(\log v)^2,$ where $\kappa(\alpha)\to 0$ as $\alpha\to 0,$ then we have

\begin{align}
  \|\mathbb{E}[g(Y)|X]\|_{1+\eta} &\leq \left(1+\eta\Ent(\chi_g(X)) +
                                    a\frac{\eta^2}{2}\kappa\left(\frac{\sqrt{2\Ent(\chi_g(X))}}{\min_x
                                    P_X(x)}\right)\right)^{\frac{1}{1+\eta}}. \label{eq:exp_upper_bound}
\end{align}

Using \eqref{eq:exp_lower_bound} and \eqref{eq:exp_upper_bound}, we
find that for any $g\in U(\epsilon_0),$ we have $\tau(g)\geq
\beta(\Ent(\chi_g(X)))$ where 

\begin{align*}
  \beta(\rho):=
  \begin{cases}
    1 & \mbox{ if } a\kappa\left(\frac{\sqrt{2\rho}}{\min_x
                                    P_X(x)}\right) - (1+\delta)^2\rho^2\leq 0 \\
   \min\left\{ \frac{2\delta\rho}{a\kappa\left(\frac{\sqrt{2\rho}}{\min_x
                                    P_X(x)}\right) -
                                (1+\delta)^2\rho^2},1\right\} &
                              \mbox{ else.}
  \end{cases}
\end{align*}
Given any $\theta>0,$ there exists $0<\epsilon_1<\epsilon_0$ small
enough so that $\Ent(\chi_g(X))\leq \Ent(g(Y))\leq \theta$ for all
$g\in U(\epsilon_1).$ This means that for all $g\in U(\epsilon_1),$ we
have $\tau(g) \geq \inf_{0\leq \rho\leq \theta} \beta(\rho).$ Since
$\kappa(\alpha) = \alpha^2 + O(\alpha^3)$ for small $\alpha>0,$ it
follows that $\inf_{0\leq \rho\leq \theta} \beta(\rho)>0$ for
sufficiently small $\theta.$ This completes the proof of the lemma.

  


\end{proof}

\end{proof}

\vspace{0.1in}

Now, we present the proof of Corollary~\ref{cor:limit_p=infty}.

\vspace{0.1in}

\begin{proof}[Proof of Corollary~\ref{cor:limit_p=infty}]
  If $X$ and $Y$ are independent, then it is clear that $\rho_m(X;Y) =
  s^*(X;Y)=0$ and $q_p^*(X;Y)=1$ for all $p\neq 1.$ The claim is
  obvious in this case.

Suppose $X$ and $Y$ are not independent. Fix any $\epsilon$ satisfying
$0<\epsilon < s^*(Y;X).$ Note that by
Theorems~\ref{thm:main-inequality} and~\ref{thm:limit_p=1}, we have
$s^*(Y;X) = \lim_{p\to 1} s_p(X;Y)\geq \rho_m^2(X;Y)>0.$

From Thm.~\ref{thm:limit_p=1}, we have that there exists a
$\delta>0$ such that 
\begin{align}
  0<|p-1|\leq\delta  & \implies s^*(Y;X)-\epsilon\leq
  \frac{q_p^*(X;Y)-1}{p-1}\leq s^*(Y;X)+\epsilon.\label{eq:derivativebound}
\end{align}

Now, define 

\begin{align}
  A(\epsilon)&:= \left\{(p,q): 0<|p-1|\leq \delta,
    s^*(Y;X)+\epsilon\leq \frac{q-1}{p-1}\leq 1\right\},
  \\
  B(\epsilon)&:= \left\{(p,q): 0<|p-1|\leq \delta,
    s^*(Y;X)-\epsilon\leq \frac{q-1}{p-1}\leq
    1\right\}\cup\{(1,1)\}\nonumber \\
  & \hspace{40pt}\cup \left\{(p,q): |p-1|\geq \delta,
    \rho_m^2(X;Y)\leq \frac{q-1}{p-1}\leq
    1\right\}.
\end{align}

From \eqref{eq:derivativebound} and Thm.~\ref{thm:main-inequality}, it is
clear that 
\begin{align}
  A(\epsilon) \subseteq\mathcal{R}(X;Y)\subseteq B(\epsilon).
\end{align}

By using the duality $(p,q)\in\mathcal{R}(X;Y) \Leftrightarrow
(q^\prime,p^\prime)\in\mathcal{R}(Y;X)$ for $p,q\neq 1,$ we obtain

\begin{align}
  A_1(\epsilon) \subseteq\mathcal{R}(Y;X)\subseteq B_1(\epsilon),
\end{align}

where 

\begin{align}
  A_1(\epsilon)&:= \left\{(p,q): |q-1|\geq \frac{1}{\delta},
    s^*(Y;X)+\epsilon\leq \frac{q-1}{p-1}\leq 1\right\},
  \\
  B_1(\epsilon)&:= \left\{(p,q): |q-1|\geq \frac{1}{\delta},
    s^*(Y;X)-\epsilon\leq \frac{q-1}{p-1}\leq
    1\right\}\cup\{(1,1)\}\nonumber \\
  & \hspace{40pt}\cup \left\{(p,q): 0<|q-1|\leq\frac{1}{\delta},
    \rho_m^2(X;Y)\leq \frac{q-1}{p-1}\leq
    1\right\}.
\end{align}

This immediately gives
\begin{align}
  s^*(Y;X)-\epsilon \leq \lim\inf_{p\to-\infty}
  \frac{q_p^*(Y;X)-1}{p-1} \leq \lim\sup_{p\to-\infty}
  \frac{q_p^*(Y;X)-1}{p-1} \leq s^*(Y;X)+\epsilon,\\
  s^*(Y;X)-\epsilon \leq \lim\inf_{p\to\infty}
  \frac{q_p^*(Y;X)-1}{p-1} \leq \lim\sup_{p\to\infty}
  \frac{q_p^*(Y;X)-1}{p-1} \leq s^*(Y;X)+\epsilon.
\end{align}

Since this is true for each sufficiently small $\epsilon>0,$
interchanging $X$ and $Y$ completes the proof.

\end{proof}

\end{document}